\documentclass[a4paper,reqno]{amsart} 

\usepackage{amssymb}
\usepackage[mathcal]{euscript} 

\usepackage[hidelinks]{hyperref}

\usepackage{color}

\newcommand{\bydef}{:=}

\newcommand{\buno}{\mathbf{1}}
\newcommand{\bzero}{\mathbf{0}}

\newcommand{\CC}{\mathbb{C}}

\newcommand{\FF}{\mathbb{F}}

\DeclareMathOperator{\Aut}{\mathrm{Aut}}%automorphism group
\newcommand{\GL}{\mathrm{GL}}

\newtheorem{theorem}{Theorem}[section]
\newtheorem{proposition}[theorem]{Proposition}
\newtheorem*{proposition*}{Proposition}
\newtheorem{lemma}[theorem]{Lemma}
\newtheorem{corollary}[theorem]{Corollary}

\theoremstyle{definition} 
\newtheorem{df}[theorem]{Definition}
\newtheorem{example}[theorem]{Example}

\theoremstyle{remark} \newtheorem{remark}[theorem]{Remark}
\numberwithin{equation}{section}

\begin{document}

\title[On binary codes that are maximal totally isotropic subspaces]{On binary codes that are maximal\\ totally isotropic subspaces\\ with respect to an alternating form}

\author{Patrick King}

\address{Department of Mathematics, Simon Fraser University, Burnaby, BC, V5A 1S6, Canada}

\email{patrickhking@gmail.ca}

\thanks{The first author acknowledges support by Memorial University and NSERC funding from the Undergraduate Student Research Award program and Discovery Grants of Drs.~Mikhail Kotchetov and Alex Shestopaloff.
The second author is supported by NSERC Discovery Grant 2018-04883.}

\author{Mikhail Kochetov}

\address{Department of Mathematics and Statistics, Memorial University of Newfoundland, St. John’s, NL, A1C 5S7, Canada}

\email{mikhail@mun.ca} 

%51E26


\begin{abstract}
    Self-dual binary linear codes have been extensively studied and classified for length \(n\leq 40\). However, little attention has been paid to linear codes that coincide with their orthogonal complement when the underlying inner product is not the dot product. In this paper, we introduce an alternating form defined on \(\mathbb{F}_2^n\) and study codes that are maximal totally isotropic with repsect to this form. We classify such codes for \(n\leq 24\) and present a MacWilliams-type identity which relates the weight enumerator of a linear code and that of its orthogonal complement with respect to our alternating inner product. As an application, we derive constraints on the weight enumerators of maximal totally isotropic codes.
\end{abstract}

\maketitle

%\tableofcontents

%--------------------------
\section{Introduction}\label{sec:intro}

Given an alphabet $M$ which is a module over a ring $R$, a \emph{linear code of length \(n\)} is an \(R\)-submodule of \(M^n\). If \(M\) is equipped with an inner product \(\langle \cdot, \cdot \rangle_M : M \times M \to R\), then one can define \(\langle v, w \rangle = \sum_{i=1}^n \langle v_i, w_i \rangle_M\) for \(v, w \in M^n\). We are interested in the case when \(M = R = \mathbb{F}_2\), the field of two elements. In this setting, we refer to a linear code \(K \subset \mathbb{F}_2^n\) of dimension \(k\) as a \emph{binary linear \([n,k]\)-code}, or, simply, a \emph{binary linear code}. A popular choice of inner product in this case is the ordinary dot product:
\begin{equation*}
    v \cdot w \bydef \sum_{i=1}^n v_iw_i,
\end{equation*}
which allows us to define the \(dual\) of \(K\) as \(K^\perp = \{v \in \mathbb{F}_2^n \mid v \cdot w = 0 \textnormal{ for all } w \in K\}\). We say that \(K\) is \emph{self-orthogonal} if \(K \subset K^\perp\) and \emph{self-dual} if \(K = K^\perp\). Self-dual codes have been extensively studied (see e.g. the survey paper \cite{RS} or monograph \cite{NRS}) and classified up to the action of \(S_n\) for \(n \leq 40\) (see \cite{Bouy12} and \cite{Bouy15}). 

If we equip \(\mathbb{F}_2^n\) with the norm defined by \(\lVert v \rVert = \lvert \{1 \leq i \leq n \mid v_i \neq 0\} \rvert\) for $v\in\FF_2^n$, which is known as the \emph{Hamming weight of} \(v\), then the \emph{minimum distance} of \(K\) is 
\begin{equation*}
    d(K) \bydef \min \{\lVert v - w \rVert \mid v,w \in K,\, v \neq w\}.
\end{equation*}
This quantity determines the \emph{error-correcting} and \emph{error-detecting} capabilities of the given code (see \cite[p.~49]{Adam}) and, as it turns out, many self-dual codes have high minimum distance (e.g. the extended Hamming [8,4]-code), which is one of the motivations for their study. The Hamming weights of the elements of $K$, including the minimum distance, is encoded in its \emph{weight enumerator}\footnote{Some authors prefer to define the weight enumerator so that powers of \(y\) count the number of ones rather than zeros. This should be noted when comparing with other texts.}:
\begin{equation*}
    W_K(x,y) \bydef \sum_{i=0}^n A_i x^iy^{n-i} = \sum_{v \in K}x^{\Vert v \rVert}y^{n - \lVert v \rVert}
\end{equation*}
where \(A_i \bydef \lvert \{v \in K \mid \lVert v \rVert = i\} \rvert\) is the number of elements in \(K\) of Hamming weight \(i\). In the case of self-dual codes, the MacWilliams identity (see the next section) is one tool that allows us to obtain information about the weight enumerator and, in some cases, compute it without explicitly knowing the structure of the code. 

Although the literature on self-dual codes is quite extensive, little attention has been paid to linear codes that are equal to their orthogonal complement when the underlying inner product is not the dot product. In light of this, we consider the following inner product defined on \(\mathbb{F}_2^n\):
\begin{equation*}
    \langle v, w \rangle \bydef \sum_{i \neq j} v_iw_i = v \cdot w + p(v)p(w)
\end{equation*}
where \(p(v) \bydef \sum_{i = 1}^n v_i\) is the \emph{parity} of \(v\) which is alternatively given by the image of \(\lVert v \rVert\) under the quotient map \(\mathbb{Z} \to \mathbb{F}_2\). Unless otherwise stated, we always denote by \(K^\perp\) the orthogonal complement of \(K\) in \(\mathbb{F}_2^n\) with respect to this inner product. 

Like the dot product, this inner product is invariant under the action of the symmetric group $S_n$ by permuting coordinates. One difference from the dot product that we wish to highlight is that \(\langle \cdot, \cdot \rangle\) is \emph{alternating}, i.e., \(\langle v, v \rangle = 0\) for all \(v \in \mathbb{F}_2^n\). It then follows that a linear code \(K\) satisfying \(K \subset K^\perp\), i.e., \(K\) is \emph{totally isotropic}, is maximal with respect to this property if and only if \(K = K^\perp\). In the case when \(n\) is even, \(\langle \cdot, \cdot \rangle\) is nondegenerate, so it makes $\FF_2^n$ a \emph{symplectic space}; in this case the subspaces satisfying $K = K^\perp$ are known as \emph{Lagrangian subspaces} or, simply, \emph{Lagrangians}. As in the study of self-dual codes, we wish to classify, up to the action of \(S_n\), the maximal totally isotropic subspaces of \(\mathbb{F}_2^n\). As we will see, for odd $n$, this is equivalent to classifying maximally self-orthogonal codes of odd length. For even $n$, if \(L\) is a Lagrangian consisting entirely of even-weight vectors, then \(L\) is simply a self-dual code. Therefore, we are particularly interested in classifying \emph{odd Lagrangians}, i.e., Lagrangian subspaces that contain a vector of odd weight (see Proposition~\ref{prop:reduction_to_even}). We acheive a complete classification of such subspaces for \(n \leq 24\) (see Table~\ref{Table:OddLag}). Somewhat surprisingly, the minimum distance of various odd Lagrangians for some $n$ is greater than the highest minimum distance of self-dual codes for the same $n$ (in Table~\ref{Table:OddLag}, the minimum distance is highlighted when this is the case). 

Another difference we wish to highlight is that the MacWilliams identity does not hold in general for the inner product \(\langle \cdot, \cdot \rangle\). However, we propose a version of this identity which is valid for this inner product (see Proposition \ref{MW-type}) and, hence, the weight enumerator of \(K^\perp\) is still determined by the weight enumerator of \(K\). After establishing this identity, we explore the implications it has on weight enumerators of codes that are maximal totally isotropic subspaces.

The paper is structured as follows. After reviewing the background on self-dual binary codes in Section~\ref{sec:prelim}, we use their relationship with maximal totally isotropic codes to classify the latter in Section~\ref{sec:classification}.
We begin Section \ref{sec:MW-Type} by deriving a MacWilliams-type identity for our inner product and then develop its consequences for the weight enumerator of a maximal totally isotropic code (Theorem~\ref{ClassicalThmI} for odd length and Theorem~\ref{WEtheoremI} for even length). More precise information is given in the case when all even weights are multiples of $4$ (Theorem \ref{ClassicalThmII} for odd length and Theorems \ref{W_L:n=0(8)}, \ref{W_L:n=2(8)}, and \ref{W_L:n=-2(8)} for even length). The code, written using SageMath \cite{sage}, for the computations involved in the classification process can be found in Appendix \ref{code}.

%--------------------------
\section{Preliminaries}\label{sec:prelim}
In this section, we review the necessary background from the theory of binary linear codes. Consider the action of \(S_n\) on the space \(\mathbb{F}_2^n\) by permuting coordinates, i.e., \(\sigma \cdot v \bydef v_{\sigma^{-1}(1)} \cdots v_{\sigma^{-1}(n)}\). We say that binary linear codes \(K\) and \(L\) are \emph{equivalent} if \(\sigma \cdot K = L\) for some \(\sigma \in S_n\), i.e., \(K\) and \(L\) lie in the same \(S_n\)-orbit. We begin with an overview of the classification of self-dual binary codes up to this action of \(S_n\). Recall that the dual of \(K\) is its orthogonal complement in \(\mathbb{F}_2^n\) taken with respect to the dot product: \(\{v \in \mathbb{F}_2^n \mid v \cdot w = 0 \textnormal{ for all } w \in K\}\). In this section only, we denote the dual of \(K\) by \(K^\perp\). 

\subsection{Binary Self-Dual Codes} Suppose that \(K\) is a self-dual binary code of length \(n\), i.e., \(K\) is a subspace of \(\mathbb{F}_2^n\) satisfying \(K = K^\perp\). Let us note some basic properties that \(K\) possesses. First, we have \(v \cdot v = 0\) for all \(v \in K\), so \(K\) consists entirely of \emph{even vectors}, i.e., vectors whose Hamming weight is even. By the nondegeneracy of the dot product, we also have that \(n = 2 \dim K\) is even, which implies  that \(K\) contains the vector of all ones, denoted \(\textbf{1}\). We provide some examples of self-dual codes below. 

\begin{example}\label{repetition}
The simplest example of a binary self-dual code is the \emph{repetition code} \(i_2 = \{00, \, 11\}\).
\end{example}

\begin{example}\label{extendedHamming}
A less trivial example is the \emph{extended Hamming \([8,4]\)-code}, \(e_8\), which is given by the row space of the following matrix:
\begin{equation*}
    \begin{bmatrix}
        1 & 0 & 0 & 0 & 0 & 1 & 1 & 1\\
        0 & 1 & 0 & 0 & 1 & 0 & 1 & 1\\
        0 & 0 & 1 & 0 & 1 & 1 & 0 & 1\\
        0 & 0 & 0 & 1 & 1 & 1 & 1 & 0
    \end{bmatrix}.
\end{equation*}
The minimum distance of \(e_8\) is \(4\), which is the highest among all possible self-dual codes of lenth \(8\). Moreover, up to equivalence, \(e_8\) is the only self-dual code of length \(8\) with minimum distance \(4\).
\end{example}

An additional property of \(e_8\) in Example~\ref{extendedHamming} is that all its vectors have Hamming weight divisible by \(4\). We refer to binary self-dual codes with this property as \emph{type II codes}. Otherwise, the code is said to be a \emph{type I code}\footnote{Note that, under our convention, type I codes are not type II. Some texts use the term \emph{strictly type I} in this case.}.  

One of the main objectives in the theory of self-dual codes is to classify such codes up to the action of \(S_n\). A table summarizing the classification of self-dual codes for \(n \leq 38\) can be found in \cite[p.~3]{Bouy12}. For an overview of how self-dual codes are built up via ``gluing'' self-orthogonal codes, i.e., codes that are contained in their dual, see \cite[\S 11]{NRS}. An important notion which arises in such constructions is that of an \emph{indecomposable} code. These are codes which cannot be expressed as a direct product of two nontrivial codes. Another notion that is crucial in the enumeration process is that of \emph{mass formulae}. We now proceed to describe this as in \cite[p.~6]{NRS}.

Given a binary linear code \(C\), consider the subgroup of \(S_n\) consisting of permutations that fix \(C\) as a set. This group is called the \emph{automorphism group of \(C\)} and denoted by \(\Aut (C)\). We then see that the number of codes that are equivalent to \(C\) is 
\begin{equation*}
   \frac{\lvert S_n \rvert}{\lvert \Aut (C) \rvert} = \frac{n!}{\lvert \Aut (C) \rvert}.
\end{equation*}
Thus, if \(S\) is a set of representatives of all equivalence classes of self-dual codes of length $n$, then the total number of self-dual codes, \(T_n\), is given by 
\begin{equation*}
    T_n = \sum_{C \in S} \frac{n!}{\lvert \Aut (C) \rvert}.
\end{equation*}
We then obtain the following mass formula for binary self-dual codes:
\begin{equation}\label{massFormula}
    \frac{T_n}{n!} = \sum_{C \in S} \frac{1}{\lvert \Aut (C) \rvert}.
\end{equation}

Equation \eqref{massFormula} is useful for enumerating binary self-dual codes up to equivalence, because it allows us to verify whether our enumeration is complete. Indeed, after computing the orders of the automorphism groups, we may compare the sum in the right-hand side with $T_n/n!$, since $T_n$ is known:
\begin{equation}\label{eq:Tn}
    T_n = \prod_{i=1}^{\frac{n}{2} - 1}(2^i + 1),
\end{equation}
see e.g. \cite[p.~8]{NRS}.

\subsection{The MacWilliams Identity}\label{MWSection} Let \(K \subset \mathbb{F}_2^n\) be a binary linear code and set \(A_i = \lvert \{v \in \mathbb{F}_2^n \mid \lVert v \rVert = i\} \rvert\). The sequence \((A_0, \dots, A_n)\) is called the \emph{weight distribution of \(K\)}. Recall that we define the weight enumerator of \(K\) to be the generating function of its weight distribution and express it as a polynomial in two indeterminates as follows:
\begin{equation*}
    W_K(x,y) = \sum_{i=0}^n A_ix^iy^{n-i} = \sum_{v \in K}x^{\Vert v \rVert}y^{n - \lVert v \rVert}. 
\end{equation*}
Note that if \(K\) is self-dual, then \(\textbf{1} \in K\) implies that \(W_K(x,y)\) is \emph{symmetric}, i.e., \(W_K(x,y) = W_K(y,x)\). For example, the weight enumerators of \(i_2\) in Example \ref{repetition} and \(e_8\) in Example \ref{extendedHamming} are given by \(W_{i_2} = y^2 + x^2\) and \(W_{e_8} = y^8 + 14x^4y^4 + x^8\) respectively.

A fundamental identity relating the weight enumerator of \(K\) and its dual is the following result due to MacWilliams \cite{MW} (see also \cite[\S8.7]{Adam}).
\begin{theorem}\label{MacWilliams}
    If \(K\) is a binary linear code, then
    \begin{equation*}
        W_{K^\perp}(x,y) = \frac{1}{\lvert K \rvert} W_K(y-x, y+x).
    \end{equation*}
\end{theorem}
If \(K\) is self-dual, then Theorem \ref{MacWilliams} says that
\begin{equation}\label{MWself-dual}
    W_K(x,y) = W_K\biggl( \frac{y-x}{\sqrt{2}}, \frac{y+x}{\sqrt{2}} \biggr).
\end{equation}
It is then useful to view \(W_K(x,y)\) as an element of the ring $\CC[x,y]$ of polynomial functions \(\mathbb{C}^2 \to \mathbb{C}\) on which the general linear group, \(\GL_2(\mathbb{C})\), acts in the usual way, i.e., \((g \cdot W_K)(u) \bydef W_K(g^{-1}u)\) for all \(g \in \GL_2(\mathbb{C}),\, u \in \mathbb{C}^2\). 

Since \(K\) consists entirely of even vectors, it follows from equation \ref{MWself-dual} that \(W_K\) is an invariant of the subgroup of $\GL_2(\CC)$ generated by
\begin{equation*}
    A = \frac{1}{\sqrt{2}}
    \begin{bmatrix}
        1 & 1 \\
        -1 & 1
    \end{bmatrix}
    \textnormal{ and }
    X = 
    \begin{bmatrix}
        -1 & 0 \\
        0 & 1
    \end{bmatrix}.
\end{equation*}
This group is isomorphic to the dihedral group of order \(16\) which we denote by \(D_8\). The algebra of invariants of \(D_8\) is known to be \(\mathbb{C}[W_{i_2},\, x^2y^2(y^2 - x^2)]\) (see \cite[\S7.1]{NRS}) and, therefore, \(W_K\) is a polynomial in \(W_{i_2}\) and \(x^2y^2(x^2-y^2)^2\).

Suppose now that \(K\) is a type II code. In addition to being a \(D_8\)-invariant, we also have that \(W_K\) is invariant under the action of
\begin{equation*}
    B = 
    \begin{bmatrix}
        i & 0 \\
        0 & 1
    \end{bmatrix}.
\end{equation*}
The group \(G\) generated by \(A\) and \(B\) is of order \(192\) (see \cite[\S6.2]{NRS}) and its algebra of invariants is \(\mathbb{C}[W_{e_8},\, x^4y^4(x^4 - y^4)^4]\) (see \cite[\S6.1]{NRS}) implying that \(W_K\) is a polynomial in \(W_{e_8}\) and \(x^4y^4(x^4 - y^4)^4\).

The MacWilliams identity can also be used to derive constraints that weight enumerators of maximally self-orthogonal codes of odd length must satisfy, as was done by Mallows and Sloane in \cite{MS}.

\begin{theorem}\label{self-orthogonalI}
Let $K$ be a maximally self-orthogonal binary code of odd length $n$.

    (A) Let \(a = y^7 + 7x^4y^3\) and \(R = \mathbb{C}[W_{i_2},\, x^2y^2(y^2 - x^2)]\). Then \(W_K\) lies in \(yR \oplus aR\), the free \(R\)-module generated by \(y\) and \(a\).

    (B) Suppose in addition that all vectors in \(K\) have Hamming weight divisible by \(4\) and let \(u_1 = y^{17} + 17x^4y^{13} + 187x^8y^9 + 51x^{12}y^5,\, u_2 = y^{23} + 506x^8y^{15} + 1288x^{12}y^{11} + 253x^{16}y^7\), and \(R' = \mathbb{C}[W_{e_8},\, x^4y^4(x^4 - y^4)^4]\). Then, 
    \begin{itemize}
        \item[(i)] Either \(n \equiv 1 \ (\mathrm{mod}\ 8)\) or \(n \equiv -1 \ (\mathrm{mod}\ 8)\).
        \item[(ii)] If \(n \equiv 1 \ (\mathrm{mod}\ 8)\), then \(W_K \in yR'\oplus u_1R'\).
        \item[(iii)] If \(n \equiv -1 \ (\mathrm{mod}\ 8)\), then \(W_K \in aR'\oplus u_2R'\).
    \end{itemize}
\end{theorem}

We will provide an alternative and more elementary proof of Theorem \ref{self-orthogonalI} in Section \ref{sec:MW-Type}, see Theorem \ref{ClassicalThmI} for part (A) and Theorem \ref{ClassicalThmII} for part (B).

%-------------------------- 
\section{Classification of maximal totally isotropic subspaces}\label{sec:classification}

Now we turn our attention to the alternating inner product 
\begin{equation}\label{eq:inner_prod}
\langle u,v \rangle=u\cdot v+p(u)p(v)
\end{equation}
on the vector space $V=\FF_2^n$. From now on, $K^\perp$ will refer to the orthogonal complement of $K\subset V$ with respect to this inner product. The following notation will be useful:
\[
V^+ \bydef \{v\in V\mid p(v)=0\}\text{ and }K^+\bydef K\cap V^+.
\]
As before, we will denote the vector of all ones by $\buno$.

\subsection{Odd Length}

If $n$ is odd then the inner product \eqref{eq:inner_prod} is degenerate, since $\buno$ belongs to its radical $V^\perp$. In fact, the radical is spanned by $\buno$, since $V=V^+\oplus\langle\buno\rangle$ and the inner product on $V^+$ (which coincides with the dot product) is nondegenerate. Since any maximal totally isotropic subspace $K$ of $V$ contains the radical, we have $K=K^+\oplus\langle\buno\rangle$ where $K^+$ is a Lagrangian subspace of the $(n-1)$-dimensional symplectic space $V^+$. 

On the other hand, for any subspace $L\subset V^+$, the orthogonal complement with repsect to the inner product \eqref{eq:inner_prod} is the same as with respect to the dot product, i.e., the dual code of $L$. Thus, $L$ is a Lagrangian subspace of $V^+$ if and only if $L$ is a maximally self-orthogonal code. We have proved the following:

\begin{proposition}\label{prop:odd_n}
    For odd $n$, the mapping $L\mapsto L^\perp=L\oplus\langle\buno\rangle$ is a bijection between the set of maximally self-orthogonal binary codes of length $n$ and the set of maximal totally isotropic subspaces of $\FF_2^n$ with respect to the inner product \eqref{eq:inner_prod}.\qed 
\end{proposition}

\begin{example}\label{ex:Hamming}
The Hamming $[7,4]$-code $H$ with parity bits in positions $1$, $2$ and $4$ can be defined as the row space of the following row-equivalent matrices:
\[
\begin{bmatrix}
1 & 1 & 1 & 0 & 0 & 0 & 0\\
1 & 0 & 0 & 1 & 1 & 0 & 0\\
0 & 1 & 0 & 1 & 0 & 1 & 0\\
1 & 1 & 0 & 1 & 0 & 0 & 1
\end{bmatrix}\sim
\begin{bmatrix}
1 & 0 & 0 & 0 & 0 & 1 & 1\\
0 & 1 & 0 & 0 & 1 & 0 & 1\\
0 & 0 & 1 & 0 & 1 & 1 & 0\\
0 & 0 & 0 & 1 & 1 & 1 & 1
\end{bmatrix}.
\]
It is a maximal totally isotropic subspace and has minimum distance $3$. 
The corresponding maximally self-orthogonal code is its even part $H^+$ (which coincides with the dual code of $H$ and is sometimes denoted by $e_7$).
%It will be convenient to describe $H_0$ and $H$ in more invariant terms: if we take for the indexing set $I$ the points of the Fano plane (i.e., the projective plane over $\FF_2$), then $H_0$ consists of $\bzero$ and the complements of all lines. Here the point with projective coordinates $(\alpha:\beta:\gamma)$ corresponds to the integer  with binary representation $\alpha\beta\gamma$.
\end{example}

\subsection{Even Length}

The situation is very different for even $n$. To begin with, the inner product \eqref{eq:inner_prod} is nondegenerate. 
Indeed, the orthogonal complement of $V^+$ is spanned by $\buno$, but $\buno\notin V^\perp$, hence $V^\perp$ is trivial. 
Moreover, $\langle\buno\rangle$ is the radical of the inner product on $V^+$ and, hence,  the Lagrangians of $V$ that are contained in $V^+$ must contain $\langle\buno\rangle$ and are in bijection with the Lagrangians in the $(n-2)$-dimensional symplectic space $V^+/\langle\buno\rangle$. 

We will call a Lagrangian $L$ of $V$ \emph{even} if it is contained in $V^+$, and otherwise we will call it \emph{odd}. Clearly, the even Lagrangians are the same as the self-dual binary linear codes of length $n$, so the previous observation can be used to count the latter, see Equation \eqref{eq:Tn}, since the number of Lagrangians in a \(2m\)-dimensional symplectic space over \(\mathbb{F}_q\) is given by
\begin{equation*}
    \prod_{i=1}^m(q^i+1).
\end{equation*}
The odd Lagrangians are described by the following result.

\begin{proposition}\label{prop:reduction_to_even}
Let $V=\FF_2^n$ where $n$ is even.
If $K$ is an even Lagrangian in $V$ and $\xi$ is an odd vector in $V$, then 
$(K\cap\langle\xi\rangle^\perp)\oplus\langle\xi\rangle$ is an odd Lagrangian.
Conversely, any odd Lagrangian $L$ can be obtained in this way starting from a unique even Lagrangian, namely, 
$K=L^+\oplus\langle\buno\rangle$, and any odd vector $\xi\in L$. 
Moreover, for any even Lagrangian $K$ and any complement $K_0$ for $\langle\buno\rangle$ in $K$, 
there are exactly two odd Lagrangians $L$ such that $K_0=L^+$, and these are mapped to each other by the transvection 
of $V$ defined by $\buno$, i.e., the mapping $\tau(v) \bydef v+p(v)\buno$.
\end{proposition}

\begin{proof}
If $K$ is an even Lagrangian and $\xi\in V$ is odd, then $\xi\notin K=K^\perp$, so $K_0 \bydef K\cap\langle\xi\rangle^\perp$ has codimension $1$ in $K$. Hence, $K_0\oplus\langle\xi\rangle$ is a totally isotropic subspace of the same dimension as $K$ and, therefore, an odd Lagrangian.

Conversely, if $L$ is an odd Lagrangian, then $L^+\bydef L\cap V^+$ has codimension $1$ in $L$ and does not contain $\buno$, so $K \bydef L_0\oplus\langle\buno\rangle$ is an even Lagrangian. For any odd $\xi\in L$, we have $K\cap\langle\xi\rangle^\perp=L^+$ and $L=L^+\oplus\langle\xi\rangle$. If $K'$ is an even Lagrangian such that $L=(K'\cap\langle\xi\rangle^\perp)\oplus\langle\xi\rangle$, then $K'\cap\langle\xi\rangle^\perp$ is contained in $L^+$ and, hence, must be equal to $L^+$ by dimension count. This forces $K'=L^+\oplus\langle\buno\rangle$.

Moreover, if $K$ is an even Lagrangian and $K_0$ is a complement for $\langle\buno\rangle$ in $K$, then $K_0$ is the kernel of some linear function $K\to\FF_2$, which can be extended to $V$ and, therefore, is equal to $\langle\xi,\cdot\rangle$ for some $\xi\in V$ by the nondegeneracy of the inner product. Thus, $K_0=K\cap\langle\xi\rangle^\perp$ and, hence, $\buno\notin\langle\xi\rangle^\perp$, so $\xi$ is odd. By the first paragraph, $L \bydef K_0\oplus\langle\xi\rangle$ is an odd Lagrangian and, clearly, $K_0=L^+$. Finally, if $L'$ is an odd Lagrangian such that $K_0=L'^+$, then both $L/K_0$ and $L'/K_0$ are $1$-dimensional subspaces of $K_0^\perp/K_0$ that are different from $K/K_0$. It remains to observe that there are precisely two such subspaces (since $\dim K_0^\perp/K_0=2$), and they are swapped by our transvection (since $\tau(\xi)-\xi=\buno\notin K_0$), so either $L'=L$ or $L'=\tau(L)$.
\end{proof}

The two-to-one correspondence $L\mapsto L^+$ of Proposition~\ref{prop:reduction_to_even} between odd Lagrangians and complements for $\langle\buno\rangle$ in even Lagrangians is clearly $S_n$-equivariant. Therefore, if $\{K_i\}$ is a set of representatives of the $S_n$-orbits of even Lagrangians, then a classification of odd Lagrangians up to the action of $S_n$ can be obtained by finding, for each $i$, the $\Aut(K_i)$-orbits in the set of complements for $\langle\buno\rangle$ in $K_i$. It should be noted that the two odd Lagrangians, $L$ and $L'$, corresponding to the same complement $K_i^0$ may or may not lie in the same $S_n$-orbit. If they do, then they are swapped by an element of the stabilizer of $K_i^0$ in $\Aut(K_i)$. It follows that $\Aut(L)=\Aut(L')$ is a subgroup of index $\le 2$ in the said stabilizer. 

\begin{example}\label{ex:odd_from_ext_Hamming}
Consider the extended Hamming code $e_8$ from Example \ref{extendedHamming}. 
If we take as the indexing set for the coordinates of $\FF_2^8$ the affine $3$-space over $\FF_2$, 
then $e_8$ consists of $\bzero$, $\buno$ and the indicator functions of all affine planes.
(This is a special case of Reed-Muller codes, see e.g. \cite[Ch.~9]{Adam}.)
The automorphism group $\Aut(e_8)$ is the group of affine transformations of the $3$-space, which is isomorphic to the semidirect product $\FF_2^3\rtimes\GL_3(\FF_2)$. The complements for $\langle\buno\rangle$ in $e_8$ are parametrized by the cosets $\xi+e_8$ of vectors of odd weight. Since $\Aut(e_8)$ is a $3$-transitive subgroup of $S_8$, there is only one 
orbit, represented by $e_8\cap\langle\xi\rangle^\perp$ with $\xi=10000000$. The weight enumerators of the corresponding odd Lagrangians, $L=(e_8\cap\langle\xi\rangle^\perp)\oplus\langle\xi\rangle$ and $L'=\tau(L)$, are $y^8+xy^7+7x^4y^4+7x^5y^3$ and $y^8+7x^3y^5+7x^4y^4+x^7y$ respectively, so $L$ and $L'$ are not equivalent. 
Since $\xi$ is the only vector of weight $1$ in the coset $\xi+e_8$, the automorphism group of $L$ and $L'$ is the stabilizer of $\xi$ in $\Aut(e_8)$, which is isomorphic to $\GL_3(\FF_2)$.
\end{example}

%As we have already mentioned, if \(\{K_i\}\) is a set of representatives of the \(S_n\)-orbits of even Lagrangians, i.e., representatives of inequivalent self-dual codes, then a classification of inequivalent odd Lagrangians may be obtained by finding, for each \(i\), the \(\Aut(K_i)\)-orbits in the set of complements for \(\langle \mathbf{1} \rangle\) in \(K_i\).
A database of generator matrices for representatives of inequivalent self-dual codes of length \(n = 2m \leq 40\) has been compiled by Harada and Munemasa \cite{database}. The generator matrices provided there have the property that their rows sum to \(\mathbf{1}\). If such a matrix has rows \(g_1,\dots, g_m\), then any of the \(m-1\) rows, say \(g_2,\dots, g_m\), span a complement for \(\langle \mathbf{1} \rangle\) in \(\langle g_1,\dots,g_m \rangle\). The remaining \(2^{m-1} - 1\) complements are given by \(\langle a_2 \mathbf{1} + g_2,\dots, a_m \mathbf{1} + g_m \rangle\) for each nonzero \(a_2\cdots a_m \in \mathbb{F}_2^{m-1}\). As we compute each complement, we may check if it is equivalent to any previously computed complement and, if so, disregard it. With this in mind, we give a description below of the algorithm used to obtain Table \ref{Table:OddLag}. The code, written using SageMath \cite{sage}, that implements this algorithm is included in Appendix \ref{code}.

\begin{itemize}
    \item Import the list \(SD\) of generator matrices of inequivalent self-dual codes obtained from \cite{database}.
    \item For each code \(K\) in \(SD\), generate all inequivalent complements for \(\langle \mathbf{1} \rangle\) in \(K\) according to the above paragraph.
    \item For each code \(K\) in \(SD\) and each inequivalent complement \(K_0\) for \(\langle \mathbf{1} \rangle\) in \(K\), find an odd \(\xi \in K_0^\perp\).
    \item Check if \(L = K_0 \oplus \langle \xi \rangle\) and \(L' = K_0 \oplus \langle \xi + \mathbf{1}\rangle\) are equivalent and, if so, add \(L\) to our list of inequivalent odd Lagrangians. Otherwise, add \(L\) and \(L'\).
\end{itemize}

To verify whether we have exhausted all inequivalent classes of odd Lagrangians, we used a mass formula analogous to that of Equation \eqref{massFormula} for self-dual codes. For each self-dual code \(K\) of length \(n = 2m\), there are \(2^{m-1}\) complements for \(\langle \mathbf{1} \rangle\) in \(K\) and, by Proposition \ref{prop:reduction_to_even}, each complement corresponds to exactly two odd Lagrangians. Recall from Equation \eqref{eq:Tn} that the number of binary self-dual codes of length \(n\) is
\begin{equation*}
    T_{n} = \prod_{i=1}^{m-1}(2^i + 1).
\end{equation*}
Then, the number of odd Lagrangians is given by \(2^mT_n\) and, hence, if \(S\) is a set of representatives of all equivalence classes of odd Lagrangians, then the mass formula states
\begin{equation*}
    \frac{2^mT_n}{n!} = \sum_{C \in S} \frac{1}{\lvert \Aut(C)\rvert}.
\end{equation*}

%\begin{remark}
%Note that \(2^mT_n + T_n = (2^m+1)T_n\) agrees with the usual formula for the number of Lagrangians. 
%\end{remark}

We now give a description of Table \ref{Table:OddLag}, starting with a definition that extends the concepts of type I and type II self-dual binary codes and will be useful in Section~\ref{sec:MW-Type}:
\begin{df}
A code \(C\) is of \emph{type II} if every vector in \(C^+\) has weight divisible by \(4\). 
If this is not the case, \(C\) is said to be of \emph{type I}.
\end{df}
In Table \ref{Table:OddLag}, the entries of the columns labeled by \(\#_I\) and \(\#_{II}\) indicate the number of type I and type II Lagrangians, respectively, for a particular length \(n\). The superscripts \(odd\) and \(even\) always indicate if the corresponding statistic is for odd Lagrangians or even Lagrangians. The entries of the columns labeled by \(d_{max}\) indicates the highest minimum distance that a Lagrangian achieves. When \(d_{max}^{odd}\) outperforms \(d_{max}^{even}\), we write its value in boldface. The statistics for even Lagrangians (self-dual codes) are taken from \cite{Bouy12}.

% include extra three columns between n and #_I for #_I^E, #_II^E, and d_max^E
\begin{table}[ht]
    \centering
    \begin{tabular}{|c|c|c|c|c|c|c|c|c|}
        \hline
         \(n\) & \(\#_I^{odd}\) & \(\#_{II}^{odd}\) & \(d_{max}^{odd}\) & \(\#_{max,I}^{odd}\) & \(\#_{max,II}^{odd}\) & \(\#_I^{even}\) & \(\#_{II}^{even}\) & \(d_{max}^{even}\) \\ \hline
         2 & 0 & 1 & 1 & 0 & 1 & 1 & 0 & 2 \\ 
         4 & 1 & 0 & 1 & 1 & 0 & 1 & 0 & 2 \\
         6 & 1 & 1 & \textbf{3} & 0 & 1 & 1 & 0 & 2 \\
         8 & 2 & 2 & 3 & 0 & 1 & 1 & 1 & 4 \\
         10 & 5 & 2 & \textbf{4} & 0 & 1 & 2 & 0 & 2 \\
         12 & 11 & 0 & 3 & 3 & 0 & 3 & 0 & 4 \\
         14 & 17 & 4 & 4 & 1 & 1 & 4 & 0 & 4 \\
         16 & 32 & 8 & 4 & 2 & 2 & 5 & 2 & 4 \\
         18 & 76 & 10 & 4 & 7 & 7 & 9 & 0 & 4 \\
         20 & 194 & 0 & \textbf{5} & 1 & 0 & 16 & 0 & 4 \\
         22 & 474 & 27 & \textbf{7} & 0 & 1 & 25 & 0 & 6 \\
         24 & 1439 & 95 & 7 & 0 & 1 & 46 & 9 & 8 \\
         \hline
    \end{tabular}
    \vspace{6pt}
    \caption{Inequivalent Lagrangians of length \(n \leq 24\)}
    \label{Table:OddLag}
\end{table}

We also wish to remark that, during our computations, we found that the smallest length for which two inequivalent odd Lagrangians have the same weight distribution is \(16\).

\iffalse
Their generator matrices and weight distributions are given by
\begin{equation*}
    \begin{bmatrix}
        1& 0& 1& 1& 1& 1& 0& 0& 0& 1& 1& 0& 0& 0& 1& 0 \\
        1& 1& 0& 1& 1& 1& 0& 0& 0& 1& 1& 1& 0& 0& 0& 0 \\
        0& 0& 0& 1& 0& 0& 1& 1& 0& 0& 0& 1& 1& 0& 1& 0 \\
        0& 0& 0& 0& 1& 0& 1& 0& 1& 0& 0& 1& 1& 0& 1& 0 \\
        0& 0& 0& 0& 0& 1& 0& 1& 1& 0& 0& 1& 1& 0& 1& 0 \\
        0& 0& 0& 0& 0& 0& 0& 0& 0& 1& 0& 1& 0& 1& 1& 0 \\
        0& 0& 0& 0& 0& 0& 0& 0& 0& 0& 1& 1& 0& 0& 1& 1 \\
        1& 0& 0& 0& 0& 0& 0& 0& 0& 0& 0& 1& 0& 0& 1& 0
    \end{bmatrix}
\end{equation*}

\fi

\subsection{Direct Product Decomposition}

Consider $V'=\FF_2^{n'}$, $V''=\FF_2^{n''}$ and $V=V'\times V''=\FF_2^n$ where $n=n'+n''$. In the case of the dot product, if $K'\subset V'$ and $K''\subset V''$ are self-orthogonal (respectively, self-dual) codes then so is the direct product $K'\times K''\subset V$, because $(u',u'')\cdot(v',v'')=u'\cdot v'+u''\cdot v''$ for all $u',v'\in V'$ and $u'',v''\in V''$.
In the case of the inner product defined by Equation~\eqref{eq:inner_prod}, we have the following relationship:
\begin{equation}\label{eq:decomp_inner_prod}
\langle(u',u''),(v',v'')\rangle=\langle u',v'\rangle+\langle u'',v''\rangle+\begin{vmatrix}
p(u') & p(v')\\
p(u'') & p(v'')
\end{vmatrix}.
\end{equation}
It follows that, for subspaces $L'\subset V'$ and $L''\subset V''$, the direct product $L'\times L''$ is totally isotropic if and only if both $L'$ and $L''$ are totally isotropic and at least one of them consists entirely of even vectors. Hence, by dimension count, we obtain the following possibilities for a direct product decomposition $L=L'\times L''$ of a maximal totally isotropic subspace $L\subset V$:
\begin{itemize}
    \item If $n$ is odd, then $L'$ and $L''$ must be maximal totally isotropic and the one that has even length must be even.
    \item If $n$ is even and $L$ is even, then $L'$ and $L''$ must be even and maximal totally isotropic (hence $n'$ and $n''$ are even).
    \item If $n$ is even and $L$ is odd, then either (i) $n'$ and $n''$ are even, $L'$ and $L''$ are maximal totally isotropic, and one of $L'$ and $L''$ is even and the other odd, or (ii) $n'$ and $n''$ are odd, and one of $L'$ and $L''$ is maximal totally isotropic (hence odd) and the other is maximal among totally isotropic even subspaces.
\end{itemize}

Note that in all cases at least one of the elements $e' \bydef (\buno',\bzero'')$ or $e'' \bydef (\bzero',\buno'')$ belongs to $L$. It is convenient to introduce coordinate-wise product on $V=\FF_2^n$, which makes it a Boolean algebra isomorphic to the algebra of all subsets of $\{1,\ldots,n\}$. Then the multiplication by $e'$ and $e''$ gives the projection to $V'$ and $V''$, respectively. Also, the dot product is given by $u\cdot v=p(uv)$ and, hence, our inner product can be written as 
\[
\langle u,v\rangle = p(uv)+p(u)p(v).
\]
In particular, a subspace $L\subset V$ is totally isotropic if and only if $p(uv)=p(u)p(v)$ for all $u,v\in L$. 
Suppose $L=L'\times L''$ as above and $e'\in L$. Then, for any elements $u=(u',u'')$ and $v=(v',v'')$ of $L$, we have 
\begin{equation*}
\begin{split}
\langle u',v'\rangle &= p(u'v')+p(u')p(v')=p(e'uv)+p(e'u)p(e'v)\\
&=p(e'uv)+p(e')^2p(u)p(v)=p(e'uv)+p(e')p(uv)=\langle e',uv\rangle.
\end{split}
\end{equation*}
Since $L'$ is totally isotropic, we conclude that $\langle e',uv\rangle=0$ for all $u,v\in L$. 

\begin{proposition}\label{prop:decompL}
Let $L\subset\FF_2^n$ be a maximal totally isotropic subspace. Then $L$ is decomposable (in the sense that, up to a permutation of coordinates, $L=L'\times L''$ where $L'$ and $L''$ have nonzero length) if and and only if there exists a vector $e\neq\bzero,\buno$ such that $\langle e,uv\rangle=0$ for all $u,v\in L$.
\end{proposition}

\begin{proof}
We have already proved that if $L$ is decomposable, then there exists such $e$. 
For the converse, we may assume that $e=(\buno',\bzero'')\in\FF_2^{n'}\times\FF_2^{n''}$ where $n',n''\ne 0$. 
Let $L'$ and $L''$ be the projections of $L$. 

First, observe that, for any $u\in L$, we have $\langle e,u\rangle=\langle e,u^2\rangle=0$, so $e\in L^\perp=L$. The above computation then shows that $\langle u',v'\rangle=\langle e,uv\rangle=0$ for all $u,v\in L$, so $L'$ is totally isotropic.

Since $e'\in L$, we have $p(u')=p(eu)=p(e)p(u)$ for all $u\in L$. If $n'$ is even, we get $p(u')=0$, so $L'$ consists entirely of even vectors. If $n'$ is odd, then $p(u')=p(u)$, so $L''$ consists entirely of even vectors.
In both cases, the determinant in Equation \eqref{eq:decomp_inner_prod} vanishes. Since $L$ and $L'$ are totally isotropic, this implies that $L''$ is totally isotropic, too.

Finally, we have $L\subset L'\times L''$ and $L'\times L''$ is totally isotropic by the previous paragraph. By maximality of $L$, we conclude that $L=L'\times L''$.
\end{proof}

As an application of Proposition~\ref{prop:decompL}, consider maximal totally isotropic codes of minimum distance $1$ or $2$.

Suppose $L$ contains a vector of weight $1$, say, $e\in L$ has $1$ in position $i$ and $0$ elsewhere. Then $\langle e,u\rangle=p(u)+u_i$ for all $u\in\FF_2^n$. In particular, $u_i=p(u)$ for all $u\in L$ and, hence, $\langle e,uv\rangle=p(uv)+u_i v_i=p(u)p(v)+u_i v_i=0$ for all $u,v\in L$. By Proposition~\ref{prop:decompL}, $L$ decomposes as the direct product of $\FF_2$ and either an even Lagrangian of $\FF_2^{n-1}$ (odd $n$) or the even part of a maximal totally isotropic subspace of $\FF_2^{n-1}$ (even $n$).

\begin{corollary}\label{cor:split_off_singleton}
Let $L\subset\FF_2^n$ be a maximal totally isotropic subspace such that $d(L)=1$. 
Then $L$ decomposes as $i_1\times M$, $i_2'\times M$ or $i_1\times M^+$
where $i_1=\{0,1\}$, $i_2'=\{00,10\}$, and $M$ is a maximal totally isotropic code with $d(M)\ge 2$ that is even in the first two cases.
\end{corollary}

For instance, $L$ in Example \ref{ex:odd_from_ext_Hamming} decomposes as $i_1\times e_7$ where $e_7$ is the even part of the Hamming $[7,4]$-code (see Example \ref{ex:Hamming}).

Now, suppose that $L$ contains a vector of weight $2$, say, $e\in L$ has $1$ in positions $i\ne j$ and $0$ elsewhere. Then $L$ does not separate points $i$ and $j$ in the sense that $u_i=u_j$ for all $u\in L$. We also have $\langle e,uv\rangle=u_i v_i+u_j v_j=0$. Hence, Proposition~\ref{prop:decompL} allows us to split off $i_2=\{00,11\}$ from $L$. Repeating this procedure, we get:

\begin{corollary}\label{cor:split_off_Boolean}
Let $L\subset\FF_2^n$ be a maximal totally isotropic subspace such that $d(L)=2$. 
Then $L$ decomposes as $i_2^k \times M$ where $k$ is a positive integer and $M$ is a maximal totally isotropic code with $d(M)\ge 3$.
\end{corollary}

%--------------------------
\section{A MacWilliams-Type Identity and its Consequences}\label{sec:MW-Type} 

Recall that the MacWilliams identity, Theorem \ref{MacWilliams}, says that if \(K\) is a binary linear code, then the weight enumerator of its dual is given by
\begin{equation*}
    \frac{1}{\lvert K \rvert} W_K(y-x, y+x).
\end{equation*}
If we replace the dual of \(K\) with its orthogonal complement with respect to our inner product, \(K^\perp = \{v \in \mathbb{F}_2^n \mid \langle v, w \rangle = 0 \textnormal{ for all } w \in K\}\), then this relationship does not hold in general as the following example illustrates.

\begin{example}\label{MWFails}
    Let \(K = \{0000,\, 1000\}\). Then the weight enumerators of \(K\) and \(K^\perp\) are given by \(W_K(x,y) = y^4 + xy^3\) and \(W_{K^\perp}(x,y) = y^4 + xy^3 + 3x^2y^2 + 3x^3y\) respectively. However,
    \begin{equation*}
        \frac{1}{\lvert K \rvert} W_K(y-x, y+x) = y^4 + 3xy^3 + 3x^2y^2 + x^3y.  
    \end{equation*}
\end{example}

In Example \ref{MWFails}, one may notice that the only difference between the weight enumerator of the dual of \(K\) and \(W_{K^\perp}\) is that the variables in the odd part have been swapped. Indeed, this is always the case, as we now proceed to show.

\subsection{The MacWilliams-Type Identity} Let \(K\) be a binary linear code of length \(n\). Recall that we define \(K^+\) to be the subspace of even vectors in \(K\). Furthermore, let \(K^- = K \smallsetminus K^+\), the subset of odd vectors in \(K\), and denote by \(W_K^+\) and \(W_K^-\) the even and odd parts of \(W_K\) respectively. That is,
\begin{align*}
    W_K^+(x,y) &= \frac{1}{2}\bigl[W_K(x, y) + W_K(-x,y)\bigr] = \sum_{v \in K^+}x^{\lVert v \rVert}y^{n - \lVert v \rVert} \\
    W_K^-(x,y) &= \frac{1}{2}\bigl[W_K(x, y) - W_K(-x,y)\bigr] = \sum_{v \in K^-}x^{\lVert v \rVert}y^{n - \lVert v \rVert}
\end{align*}
and so \(W_K = W_K^+ + W_K^-\). We then have the following MacWilliams-type identity for our inner product.

\begin{proposition}\label{MW-type}
    The weight enumerator of \(K^\perp\) satisfies
    \begin{equation*}
        W_{K^\perp}(x,y) = \frac{1}{\lvert K\rvert}\bigl[W_K^+(y-x, y+x) + W_K^-(y+x,y-x)\bigr].
    \end{equation*}
\end{proposition}

\begin{proof}
Set \(\chi_v(w) = (-1)^{\langle v, w \rangle}\). Then \(\chi_v\) defines a character \(K \to \{-1, 1\}\) for each \(v \in \mathbb{F}_2^n\). If \(v \in K^\perp\), then \(\chi_v\) is the trivial character. Otherwise, we have \(\sum_{w \in K} \chi_v(w) = 0\). Therefore, the indicator function on \(K^\perp\) assumes the form
\begin{equation*}
    v \mapsto \frac{1}{\lvert K \rvert}\sum_{w \in K} \chi_v(w).
\end{equation*}
Also, it can be seen via induction on \(n\) that
\begin{equation*}
    \sum_{v \in \mathbb{F}_2^n} (-1)^{v \cdot w}x^{\lVert v \rVert}y^{n - \lVert v \rVert} = \prod_{i=1}^n[y + (-1)^{w_i}x]. 
\end{equation*}
Therefore,
\begin{align*}
    &W_{K^\perp}(x, y) = \sum_{w \in K^\perp}x^{\lVert w \rVert}y^{n-\lVert w \rVert} \\
    &= \sum_{v \in \mathbb{F}_2^n} \biggl(\frac{1}{\lvert K \rvert}\sum_{w \in K} \chi_v(w)\biggr)x^{\lVert v \rVert}y^{n-\lVert v \rVert} \\
    &= \frac{1}{\lvert K \rvert} \biggl[\sum_{w \in K^+}\sum_{v \in \mathbb{F}_2^n} \chi_v(w)x^{\lVert v \rVert}y^{n-\lVert v \rVert} + \sum_{w \in K^-}\sum_{v \in \mathbb{F}_2^n} \chi_v(w)x^{\lVert v \rVert}y^{n-\lVert v \rVert}\biggr] \\
    &= \frac{1}{\lvert K \rvert} \biggl[\sum_{w \in K^+}\sum_{v \in \mathbb{F}_2^n} (-1)^{v \cdot w}x^{\lVert v \rVert}y^{n-\lVert v \rVert} + \sum_{w \in K^-}\sum_{v \in \mathbb{F}_2^n} (-1)^{v \cdot w}(-x)^{\lVert v \rVert}y^{n-\lVert v \rVert}\biggr] \\
    &= \frac{1}{\lvert K \rvert} \biggl[\sum_{w \in K^+}\prod_{i=1}^n[y + (-1)^{w_i}x] + \sum_{w \in K^-}\prod_{i=1}^n[y + (-1)^{w_i}(-x)]\biggr].
\end{align*}
The following observation completes the proof:
\begin{align*}
    W_K^+ (y-x, y+x) &= \sum_{w \in K^+}(y-x)^{\lVert w \rVert}(y+x)^{n - \lVert w \rVert} = \sum_{w \in K^+} \prod_{i = 1}^n[y + (-1)^{w_i}x] \\
    W_K^- (y+x, y-x) &= \sum_{w \in K^+}(y+x)^{\lVert w \rVert}(y-x)^{n - \lVert w \rVert} = \sum_{w \in K^-} \prod_{i = 1}^n[y + (-1)^{w_i}(-x)].
\end{align*}
\end{proof}

\begin{remark}
    The formula in Proposition \ref{MW-type} agrees with the ordinary MacWilliams identity, Theorem \ref{MacWilliams}, if and only if \(W_K^-\) is symmetric. In particular, if \(K\) is even, then the two formulas agree, as expected. Similarly, the formulas agree if \(K^\perp\) is even, which can happen only for even \(n\) and is equivalent to \(W_K\) being symmetric. %or, equivalently, \(K\) containing \(\mathbf{1}\). 
    Later, we will see examples of codes that satisfy the ordinary identity, but are not in either of these classes, i.e., codes \(K\) such that \(W_K^+\) is not symmetric and \(W_K^- \neq 0\) is symmetric.
\end{remark}

If \((A_0,\dots, A_n)\) is the weight distribution of \(K\) and \((A_0',\dots, A_n')\) the weight distribution of its dual, then the original MacWilliams identity may be used to express each \(A_k'\) in terms of \(A_0,\dots, A_n\). More precisely, we have 
\begin{equation*}
    A_k' = \frac{1}{\lvert K \rvert} \sum_{i=0}^n A_i P_k(i;n)
\end{equation*}
where \(P_k(x;m) = \sum_{i=0}^k (-1)^i \binom{x}{i} \binom{m-x}{k-i}\) is a \emph{Krawtchouk polynomial} (see \cite[\S 5]{TOECC}). If \((A_0^\perp,\dots, A_n^\perp)\) is the weight distribution of \(K^\perp\), then Proposition \ref{MW-type} gives the following corollary:

\begin{corollary}
    For each \(0\leq k \leq n\), we have
    \begin{align*}
        A_k^\perp &= A_k' = \frac{1}{\lvert K \rvert} \sum_{i=0}^n A_i P_k(i;n) \textnormal{ for even } k \\
        A_k^\perp &= A_{n-k}' = \frac{1}{\lvert K \rvert} \sum_{i=0}^n A_i P_{n-k}(i;n) \textnormal{ for odd } k.
    \end{align*}
\end{corollary}

\subsection{Consequences in the Odd Length Case}\label{oddConsequences} Let \(L\) be a maximal totally isotropic subspace of odd length \(n\). Since $L^\perp=L$, Proposition \ref{MW-type} implies that the weight enumerator of \(L\) satisfies
\begin{equation}\label{MW-typeOdd}
    W_L(x,y) = \frac{1}{\lvert L \rvert}\bigl[W_L^+(y-x, y+x) + W_L^-(y+x,y-x)\bigr].
\end{equation}
Recall that \(L = L^+ \oplus \langle \textbf{1} \rangle\) and so \(W_L\) is symmetric. Since \(n\) is odd, it follows that 
\begin{equation}\label{symmetry-typeOdd}
W_L^+(x,y) = W_L^-(y,x).
\end{equation}
Setting \(n-1 = 2m\) and using Equation \eqref{MW-typeOdd}, we obtain
\begin{equation}\label{MW-typeOdd2}
    W_L(x,y) = \frac{1}{2^m}W_L^+(y-x, y+x) = \frac{1}{2^m}W_L^+(x-y, y+x).
\end{equation}
Alternatively, Equation \eqref{MW-typeOdd2} can be obtained from the ordinary MacWilliams identity, applied to $L^+$, by noting that \(L^\perp=L\) coincides with the dual of \(L^+\).

Equation \eqref{MW-typeOdd2} tells us that the weight enumerator of \(L\) is completely determined by the weight enumerator of the corresponding maximally self-orthogonal code \(L^+\). Hence, the classical results stated in Theorem \ref{self-orthogonalI} may be used to describe the weight enumerator of \(L\). We will now provide an alternative proof of Theorem \ref{self-orthogonalI}, which is more elementary and will be useful in the next subsection to establish new results for even length.  We begin with the general case and then give the improved version for codes of Type II.

\subsubsection{General Case}\label{AltProofI} Recall from Section \ref{MWSection} that the group generated by 
\begin{equation}\label{standard_rep_D8}
    A = \frac{1}{\sqrt{2}}
    \begin{bmatrix}
        1 & 1 \\
        -1 & 1
    \end{bmatrix}
    \textnormal{ and }
    X = 
    \begin{bmatrix}
        -1 & 0 \\
        0 & 1
    \end{bmatrix}.
\end{equation}
is isomorphic to \(D_8\) and that the algebra of invariants of \(D_8\) is the polynomial algebra \(\mathbb{C}[s,t]\) where \(s = W_{i_2} = x^2 + y^2\) and \(t = x^2y^2(x^2-y^2)^2\). It is well known that \(D_8\) has four \(1\)-dimensional irreducible representations and three \(2\)-dimensional irreducible representations (see e.g. \cite[\S18.3]{JL1993}). %We keep this in mind for the following lemma.

\begin{lemma}\label{D8representationOdd}
    The subspace of \(\mathbb{C}[x,y]\) spanned by \(W_L^+\) and \(W_L^-\) is \(D_8\)-invariant and gives the standard 2-dimensional irreducible representation of \(D_8\), i.e., the representation given by the matrices \eqref{standard_rep_D8}.
    %of \(D_8\) whose corresponding character maps \(A \mapsto \sqrt{2}\) and \(X \mapsto 0\).
\end{lemma}

\begin{proof}
Set \(e_1 = W_L^+\) and \(e_2 = W_L^-\). Note that \(X \cdot e_1 = e_1\) and \(X \cdot e_2 = -e_2\). Now, from Equation \eqref{MW-typeOdd2}, we obtain
\begin{equation*}
    A \cdot e_1 = \frac{1}{\sqrt{2}}(e_1 + e_2).
\end{equation*}
Furthermore, \(A^2 \cdot e_1 = W^+_L(-y,x) = W^+_L(y,x) = e_2\) by Equation \eqref{symmetry-typeOdd}, hence 
\begin{align*}
    & e_2 = \frac{1}{\sqrt{2}}\bigl(A\cdot e_1 + A\cdot e_2\bigr) \\
    \implies &A\cdot e_2 = \frac{1}{\sqrt{2}} \bigl(-e_1 + e_2\bigr).
\end{align*}
\end{proof}

Recall that if a group \(G\) acts linearly on a vector space \(V\) and \(\chi\) is a \(1\)-dimensional character of \(G\), then \(u \in V\) is called a \emph{semi-invariant of weight \(\chi\)} if \(g \cdot u = \chi(g)u\) for all \(g \in G\). Observe from Lemma \ref{D8representationOdd} that \(yW_L^- - xW_L^+\) is invariant under \(A\) and transforms by \(-1\) under \(X\). Therefore, \(yW_L^- - xW_L^+\) is a semi-invariant whose weight \(\chi\) is the 1-dimensional character that maps \(A \mapsto 1\) and \(X \mapsto -1\). In particular, if \(a = W_{e_7} = y^7 + 7x^4y^3\) and \(b = W_{e_7}(y,x) = x^7 + 7x^3y^4\), then \(D \bydef yb - xa\) is a semi-invariant of weight \(\chi\). It turns out that it is the simplest such semi-invariant in the following sense: 

\begin{lemma}\label{semi-invariantD}
    If \(f\) is any semi-invariant of weight \(\chi\), then \(f = Dg\) for some \(D_8\)-invariant \(g \in \mathbb{C}[s,t]\).
\end{lemma}

\begin{proof}
    Note that \(D^2\) is an invariant and thus lies in \(\mathbb{C}[s,t]\). One can check that \(D^2 = t(s^4 - 16t)\). Set \(h = fD\). Then \(h^2 = f^2t(s^4-16t) \in \mathbb{C}[s,t]\) with \(f^2 \in \mathbb{C}[s,t]\) so \(h^2\) is divisible in \(\mathbb{C}[s,t]\) by \(t(s^4-16t)\). Since \(t\) and \(s^4-16t\) are irreducible and are not associates in the unique factorization domain \(\mathbb{C}[s,t]\), we have that \(h\) is divisible in \(\mathbb{C}[s,t]\) by \(t(s^4-16t)\), say \(h = t(s^4-16t)g\). Then, \(f = Dg\).
\end{proof}

Recall from Proposition \ref{prop:odd_n} that every maximally self-orthogonal code $K$ of odd length is the even part of a maximal totally isotropic subspace, namely, $L=K\oplus\langle\buno\rangle$. Therefore, in order to prove part (A) of Theorem \ref{self-orthogonalI}, it suffices to show that the statement holds for \(W_L^+\). We now proceed to show this. 

\begin{theorem}\label{ClassicalThmI}
    Set \(a = W_{e_7}(x,y) = y^7 + 7x^4y^3,\, b = W_{e_7}(y,x) = x^7 + 7x^3y^4\), and recall that \(\mathbb{C}[x,y]^{D_8} = \mathbb{C}[s,t]\) where \(s = x^2+y^2\) and \(t = x^2y^2(x^2-y^2)\). If $L$ is a maximal totally isotropic subspace of $\FF_2^n$ where \(n\) is odd, then \(W_L^+(x,y)\) lies in \(y\mathbb{C}[s,t] \oplus a\mathbb{C}[s,t]\), the free \(\mathbb{C}[s,t]\)-module generated by \(y\) and \(a\), and \(W^-_L(x,y)=W^+_L(y,x)\in x\mathbb{C}[s,t] \oplus b\mathbb{C}[s,t]\).
\end{theorem}

\begin{proof}
    Set \(e_1 = W_L^+\) and \(e_2 = W_L^-\). As we saw, \(-xe_1 + ye_2\) is a semi-invariant of weight \(\chi\), and so is \(-be_1 + ae_2\), because \(D_8\) acts on \(a\) and \(b\) in the same way as on \(y\) and \(x\),  respectively. Therefore, by Lemma \ref{semi-invariantD}, there exists \(f_1, f_2 \in \mathbb{C}[s,t]\) such that \(-xe_1 + ye_2 = Df_1\) and \(-be_1 + ae_2 = Df_2\). This gives a linear system in \(e_1\) and \(e_2\) with determinant \(D\). We then see that \(e_1 = af_1 - yf_2 \in y\mathbb{C}[s,t] \oplus a\mathbb{C}[s,t]\). To see that this sum is direct, suppose that \(yf + ag = 0\) for some \(f,g \in \mathbb{C}[s,t]\), then \((y^6 + 7x^4y^2)g = -f\) is an invariant which happens if and only if \(f = g = 0\).
\end{proof}

\subsubsection{Codes of Type II} We now impose the additional assumption that each \(v \in L^+\) has weight divisible by \(4\). Recall from Section \ref{MWSection} that the group \(G\) generated by
\begin{equation*}
    A = \frac{1}{\sqrt{2}}
    \begin{bmatrix}
        1 & 1 \\
        -1 & 1
    \end{bmatrix}
    \textnormal{ and }
    B = 
    \begin{bmatrix}
        i & 0 \\
        0 & 1
    \end{bmatrix}
\end{equation*}
has order \(192\) and its algebra of invariants is \(\mathbb{C}[s,t]\) where \(s = W_{e_8} = y^8 + 14x^4y^4 + x^8\) and \(t = x^4y^4(x^4-y^4)^4\). Just as in Section \ref{AltProofI}, we use the fact that the \(\mathbb{C}[s,t]\)-module of semi-invariants of a particular weight is cyclic to give an alternative proof of part (B) of Theorem \ref{self-orthogonalI}.

% put the explicit polynomials in the last column of table
\begin{lemma}\label{semi-invGI}
    The non-trivial \(1\)-dimensional characters of \(G\) have associated semi-invariants described by the following table:
    \begin{center}
        \begin{tabular}{|c|c|c|c|}
            \hline
             Character & Value on \(A\) & Value on \(B\) & Semi-invariant \\ \hline
             \(\chi_1\) & \(-1\) & \(-i\) & \(p = 2xy(y^2-x^2)(y^2+x^2)\) \\
             \(\chi_2\) & \(1\) & \(-1\) & \(w = x^2y^2(x^4-y^4)^2\) \\ 
             \(\chi_3\) & \(-1\) & \(1\) & \(u = -\frac{1}{2}(x^{12} - 33x^8y^4 - 33x^4y^8 + y^{12})\) \\ 
             \(\chi_4\) & \(1\) & \(-i\) & \(D_1 = pu\) \\ 
             \(\chi_5\) & \(1\) & \(i\) & \(D_2 = puw\) \\ 
             \(\chi_6\) & \(-1\) & \(i\) & \(D_3 = pw\) \\ 
             \(\chi_7\) & \(-1\) & \(-1\) & \(D_4 = uw\) \\ 
             \hline
        \end{tabular}
    \end{center}
    Moreover, each of the listed semi-invariants generates its corresponding \(\mathbb{C}[s,t]\)-module of semi-invariants.
\end{lemma}

\begin{proof}
    We omit the computations showing that \(p,\, u,\, w\) and, hence, \(D_1,\, D_2,\, D_3,\,\) and \(D_4\), are semi-invariants of appropriate weight. Let \(f\) be any semi-invariant of weight \(\chi_1\) and set \(h = f(pw) \in \mathbb{C}[s,t]\). Then, \(h^2 = (f^2w)(p^2w)\) with \(f^2w,p^2w \in \mathbb{C}[s,t]\) so \(p^2w = 4t\) divides \(h^2\) and, hence, divides \(h\) in \(\mathbb{C}[s,t]\) since \(\mathbb{C}[s,t]\) is a unique factorization domain. Say \(h = (p^2w)g\). Then, \(f = pg\) as desired. The same argument works when \(f\) is a semi-invariant of weight \(\chi_2,\, \chi_3\), or \(\chi_4\) by taking \(h = fw\) for \(f\) of weight \(\chi_2,\, h = fu\) for \(f\) of weight \(\chi_3\), and \(h = f(D_1w)\) for \(f\) of weight \(\chi_4\). Note that \(w^2 = t,\, u^2 = \frac{1}{4}(s^3 - 108t)\), and \(D_1^2w = t(s^3 - 108t)\) so, in each case, they divide \(h\) in \(\mathbb{C}[s,t]\).

    Suppose now that \(f\) is a semi-invariant of weight \(\chi_5\). Then, \(wf\) is a semi-invariant of weight \(\chi_4\) and \(f^2\) is a semi-invariant of weight \(\chi_2\) so there exists \(g_1,g_2 \in \mathbb{C}[s,t]\) such that \(wf = D_1g_1\) and \(f^2 = wg_2\). Now, \(wf^2 = (fD_1)g_1 = w^2g_2\) so either \(w^2 = t\) divides \(fD_1\) or \(g_1\) in \(\mathbb{C}[s,t]\). If \(g_1 = w^2h_1\), then \(wf = D_1(w^2h_1)\) so \(f = D_2h_1\) as desired. If \(fD_1 = w^2h_2\), then 
    \begin{align*}
        &D_2(wf) = D_2(D_1g_1) \\
        \implies & w^2(D_1f) = D_1^2wg_1 \\
        \implies &t^2h_2 = t(s^3-108t)g_1
    \end{align*}
    and thus \(t\) divides \(g_1\) in \(\mathbb{C}[s,t]\) and we have that \(f\in D_2\mathbb{C}[s,t]\) as before. It is then easy to see that if \(f\) is any semi-invariant of weight \(\chi_6\), then \(fu\) is a semi-invariant of weight \(\chi_5\) and thus lies in \(D_2\mathbb{C}[s,t] = puw\mathbb{C}[s,t]\) giving \(f \in pw\mathbb{C}[s,t] = D_3\mathbb{C}[s,t]\).

    Finally, suppose that \(f\) is a semi-invariant of weight \(\chi_7\). Then \(fu\) and \(fw\) are semi-invariants of weight \(\chi_2\) and \(\chi_3\), respectively, so there exists \(g_1,g_2 \in \mathbb{C}[s,t]\) such that \(fu = wg_1\) and \(fw = ug_2\). Then, 
    \begin{align*}
        &fu^2g_2 = fw^2g_1 \\
        \implies & (s^3 - 108t)g_2 = t(4g_1)
    \end{align*}
    so \(t\) divides \(g_2\) in \(\mathbb{C}[s,t]\), say \(g_2 = tg_3\). Hence, \(fw = utg_3\) and, thus, \(f = (uw)g_3\).
\end{proof}

We now give an alternative proof of part (B) of Theorem \ref{self-orthogonalI}.

\begin{theorem}\label{ClassicalThmII} 
    Let \(u_1 = y^{17} + 17x^4y^{13} + 187x^8y^9 + 51x^{12}y^5,\, u_2 = y^{23} + 506x^8y^{15} + 1288x^{12}y^{11} + 253x^{16}y^7,\) and \(a = W_{e_7} = y^7 + 7x^4y^3\). If \(L\) is a maximal totally isotropic subspace of \(\mathbb{F}_2^n\), where \(n\) is odd, and each vector in \(L^+\) has weight divisible by \(4\), then
    \begin{itemize}
        \item[(i)] Either \(n \equiv 1 \ (\mathrm{mod}\ 8)\) or \(n \equiv -1 \ (\mathrm{mod}\ 8)\).
        \item[(ii)] If \(n \equiv 1 \ (\mathrm{mod}\ 8)\), then \(W_L^+ \in y\mathbb{C}[s,t]\oplus u_1\mathbb{C}[s,t]\).
        \item[(iii)] If \(n \equiv -1 \ (\mathrm{mod}\ 8)\), then \(W_L^+ \in a\mathbb{C}[s,t]\oplus u_2\mathbb{C}[s,t]\).
    \end{itemize}
\end{theorem}

\begin{proof}
    Set \(e_1 = W_L^+\) and \(e_2 = W_L^-\). Recall that \(L = L^+ \oplus \langle \textbf{1} \rangle\) and thus, if \(\xi = v + \textbf{1} \in L\) is odd, then \(\Vert \xi \rVert = \lVert v \rVert + \lVert \textbf{1} \rVert - 2\lVert v\xi \rVert\) where the product \(v\xi\) is taken component-wise. Therefore, \(\lVert \xi \rVert \equiv n \ (\textrm{mod}\ 4)\). 
    
    Suppose first that \(n\equiv 1 \ (\textrm{mod}\ 4)\). Then \(B \cdot e_1 = e_1\) and \(B \cdot e_2 = -ie_2\). Recalling that 
    \begin{align*}
        A \cdot e_1 &= \frac{1}{\sqrt{2}}(e_1 + e_2) \\
        A \cdot e_2 &= \frac{1}{\sqrt{2}}(e_2 - e_1),
    \end{align*}
    we have that \(ye_2 - xe_1\) is a semi-invariant of weight \(\chi: A\mapsto 1,\, B\mapsto -i\). Hence, by Lemma \ref{semi-invGI}, we have that \(ye_2 - xe_1 = D_1f_1\) for some \(f_1 \in \mathbb{C}[s,t]\). Moreover, if \(w_1(x,y) = u_1(y,x)\), then it can be checked that \(G\) acts on \(u_1\) and \(w_1\) as on \(e_1\) and \(e_2\) respectively. It follows that \(u_1e_2 - w_1e_1 = D_1f_2\) for some \(f_2 \in \mathbb{C}[s,t]\). Since \(D_1 = yw_1 - xu_1\), we obtain, as in the proof of Theorem \ref{ClassicalThmI}, that \(e_1 = u_1f_1 - yf_2 \in y\mathbb{C}[s,t] \oplus u_1\mathbb{C}[s,t]\). Since the degrees of \(y\) and \(u_1\) are congruent to \(1\) modulo \(8\), it follows that \(n \equiv 1 \ (\textrm{mod}\ 8)\). Showing that the sum is direct can be done in the same way as in the proof of Theorem \ref{ClassicalThmI}.

    Now, suppose that \(n \equiv -1 \ (\textrm{mod}\ 4)\). Then, \(B \cdot e_1 = e_1\) and \(B \cdot e_2 = ie_2\). If \(b(x,y) = a(y,x)\) and \(w_2(x,y) = u_2(y,x)\), then it can be checked that \(G\) acts on \(a\) and \(b\) as on \(e_1\) and \(e_2\) respectively. Similarly, \(G\) acts on \(u_2\) and \(w_2\) as on \(e_1\) and \(e_2\) respectively. Hence, \(ae_2 - be_1 = D_2f_1\) and \(u_2e_2 - w_2e_1 = D_2f_2\) are both semi-invariants of weight \(\chi^{-1}\). Since \(7D_2 = aw_2 - bu_2\), we obtain \(7e_1 = u_2f_1 - af_2 \in a\mathbb{C}[s,t] \oplus u_2 \mathbb{C}[s,t]\) just as before. Since the degrees of \(a\) and \(u_2\) are congruent to \(-1\) modulo \(8\), it follows that \(n \equiv -1 \ (\textrm{mod}\ 8)\). To see that the sum is direct, suppose that there exists \(f,g \in \mathbb{C}[s,t]\) such that \(af = u_2g\). Since \(f(x,y) = f(y,x)\) and \(g(x,y) = g(y,x)\), we also have that \(bf = w_2g\). Now, 
    \begin{align*}
        7D_2f &= (aw_2 - bu_2)f \\
        &= w_2u_2g + bu_2f \\
        &= 2w_2u_2g
    \end{align*}
    is a semi-invariant of weight \(\chi^{-1}\) yet, \(A \cdot 2w_2u_2 = w_2^2 - u_2^2 \neq 2w_2u_2\).
\end{proof}

\subsection{Consequences in the Even Length Case} We now turn to the case when \(L\) is a maximal totally isotropic subspace of even length \(n\). Once again, our goal is to use Proposition \ref{MW-type} to study the weight enumerator of \(L\), starting with the general case and then specializing to Lagrangians that are of Type II. We are interested in the case when \(L\) is an odd Lagrangian because, otherwise, \(L\) is a self-dual code and we know from Section \ref{MWSection} the structure of its weight enumerator. First, we provide some examples of odd Lagrangians and their weight enumerators.

\begin{example}\label{TrivialOddLag}
    The simplest example of an odd Lagrangian is \(i_2' = \{00,\, 10\}\). Its weight enumerator is simply \(W_{i_2'} = y^2 + xy\).
\end{example}

\begin{example}\label{OddLagLength6}
    Two slightly more interesting examples of odd Lagrangians are \(L_1\) and \(L_2\) of dimension \(3\) given by the row spaces of
    \begin{equation*}
        \begin{bmatrix}
        1 & 0 & 0 & 0 & 0 & 0 \\
        0 & 1 & 0 & 1 & 0 & 0 \\
        0 & 0 & 1 & 0 & 1 & 0 
        \end{bmatrix}
        \textnormal{ and }
        \begin{bmatrix}
        1 & 0 & 0 & 0 & 1 & 1 \\
        0 & 1 & 0 & 1 & 0 & 1 \\
        0 & 0 & 1 & 1 & 1 & 0  
        \end{bmatrix},
    \end{equation*}
    repectively. They are representatives of the only two \(S_6\)-inequivalent classes of odd Lagrangians of dimension \(3\) (see Table \ref{Table:OddLag}). Their weight enumerators are
    \begin{align*}
        W_{L_1} &= y^6 + xy^5 + 2x^2y^4 + 2x^3y^3 + x^4y^2 + x^5y \\
        W_{L_2} &= y^6 + 4x^3y^3 + 3x^4y^2.
    \end{align*}
\end{example}

\subsubsection{Consequences for General Lagrangians}As in Section \ref{oddConsequences}, let \(A\) and \(X\) be the generators of \(D_8\) and \(s = x^2+y^2\) and \(t = x^2y^2(x^2-y^2)\) be the polynomials generating the algebra of \(D_8\)-invariants. By Proposition \ref{MW-type}, we have
\begin{equation}\label{MW-typeEvenI}
    W_L(x,y) = W_L^+\biggl(\frac{y-x}{\sqrt{2}}, \frac{y+x}{\sqrt{2}}\biggr) + W_L^-\biggl(\frac{y+x}{\sqrt{2}},\frac{y-x}{\sqrt{2}}\biggr)
\end{equation}
which may alternatively be written as 
\begin{equation}\label{MW-typeEven}
    W_L = A \cdot W_L^+ + A^{-1} \cdot W_L^-. 
\end{equation}
We then have the following result analogous to Lemma \ref{D8representationOdd}.

\begin{lemma}\label{D8representationEven}
    Set \(e_1 = W_L^+(x,y),\, e_2 = W_L^-(x,y),\, e_3 = W_L^+(y,x)\), and \(e_4 = W_L^-(y,x)\). Then, the subspace \(V\) of \(\mathbb{C}[x,y]\) spanned by \(e_1,\, e_2,\, e_3,\) and \(e_4\) is \(D_8\)-invariant. If \(v_1 = e_2+e_4\) and \(v_2 = e_1-e_3\), then the decomposition of \(V\) into irreducible \(D_8\)-submodules is given by \(V = \langle v_1, v_2 \rangle \oplus \langle e_2 - e_4 \rangle \oplus \langle e_1 + e_3 \rangle\).
\end{lemma}

\begin{proof}
    Note that \(e_1\) and \(e_3\) are invariant under \(X\), while \(e_2\) and \(e_4\) transform by \(-1\). Swapping the variables in Equation \eqref{MW-typeEvenI} gives
    \begin{equation*}
        e_3 + e_4 = A \cdot e_1 - A^{-1} \cdot e_2 
    \end{equation*}
    which may be added to Equation \eqref{MW-typeEven} to obtain
    \begin{equation}\label{i}
        A \cdot e_1 = \frac{1}{2}[e_1 + e_2 + e_3 + e_4]. \tag{i}
    \end{equation}
    Now, \(A^2 \cdot e_1 = W_L^+(-y,x) = e_3\). Similarly, \(A^2\cdot e_2 = -e_4,\, A^2\cdot e_3 = e_1\), and \(A^2 \cdot e_4 = -e_2\). Therefore, from Equations \eqref{MW-typeEven} and \eqref{i}, we have
    \begin{align}\label{ii}
        A \cdot e_2 &= A^2 \cdot [e_1 + e_2 - A \cdot e_1]\nonumber \\
        &= \frac{1}{2} [-e_1 + e_2 + e_3 - e_4]. \tag{ii}
    \end{align}
    Finally, from Equation \eqref{i},
    \begin{align}\label{iii}
        A \cdot e_3 &= \frac{1}{2}A^2 \cdot [e_1 + e_2 + e_3 + e_4]\nonumber \\
        &= \frac{1}{2}[e_1 - e_2 + e_3 - e_4] \tag{iii}
    \end{align}
    and, from Equation \eqref{ii},
    \begin{align}\label{iv}
        A \cdot e_4 &= -\frac{1}{2}A^2 \cdot [-e_1 + e_2 + e_3 - e_4]\nonumber \\
        &= \frac{1}{2}[-e_1 - e_2 + e_3 + e_4]. \tag{iv}
    \end{align}
    Equations \eqref{i} through \eqref{iv} show that \(V\) is \(D_8\)-invariant. We now analyze the decomposition of \(V\).

    Observe that \(e_1 + e_3\) is invariant under \(X\) and, by adding Equations \eqref{i} and \eqref{iii} together, \(A \cdot (e_1 + e_3) = e_1 + e_3\). Therefore, \(e_1 + e_3\) is an invariant and is equal to the weight enumerator of the corresponding even Lagrangian \(K = L^+ \oplus \langle \mathbf{1} \rangle\). Indeed, \(K\) is the disjoint union of \(L^+\) and \(\mathbf{1} + L^+\) and thus its weight enumerator is the sum of \(W_{L^+} = e_1\) and \(W_{\mathbf{1} + L^+} = e_3\). 
    
    Now, note that \(X \cdot (e_2 - e_4) = - (e_2 - e_4)\) and, by subtracting Equation \eqref{iv} from Equation \eqref{ii}, \(A \cdot (e_2 - e_4) = e_2 - e_4\). Hence, \(e_2 - e_4\) is a semi-invariant. 
    
    Finally, \(X \cdot v_1 = -v_1\) and \(X \cdot v_2 = v_2\). Adding Equations \eqref{ii} and \eqref{iv} together gives \(A \cdot v_1 = -v_2\) and subtracting Equation \eqref{iii} from Equation \eqref{i} gives \(A \cdot v_2 = v_1\). The subspace \(\langle v_1, v_2 \rangle\) corresponds to the unfaithful 2-dimensional irreducible representation of \(D_8\).
\end{proof}

\begin{theorem}\label{WEtheoremI}
    Let \(L\) be a Lagrangian subspace of \(\mathbb{F}_2^n\). Then,
    \begin{align*}
        W_L^+(x,y) + W_L^+(y,x) &\in \mathbb{C}[s,t] \\
        W_L^-(x,y) - W_L^-(y,x) &\in D\mathbb{C}[s,t] \\
        W_L^-(x,y) + W_L^-(y,x) &\in p_1 \mathbb{C}[s,t] \oplus q_1 \mathbb{C}[s,t] \\
        W_L^+(x,y) - W_L^+(y,x) &\in p_2 \mathbb{C}[s,t] \oplus q_2 \mathbb{C}[s,t],
    \end{align*}
    where \(p_1 = 2W_{i_2'}^-(x,y) = 2xy,\, p_2 = W_{i_2'}^+(x,y) - W_{i_2'}^+(y,x) = y^2 - x^2,\, q_1 = W_{L_1}^-(x,y) - W_{L_2}^-(x,y) = xy^5 - 2x^3y^3 + x^5y,\, q_2 = W_{L_1}^+(x,y) - W_{L_2}^+(x,y) = 2x^2y^2(y^2 - x^2)\), and  \(D = yW_{e_7}(y,x) - xW_{e_7}(x,y)\) with \(i_2',\, L_1\), and \(L_2\) as in Examples \ref{TrivialOddLag} and \ref{OddLagLength6}.
\end{theorem}

\begin{proof}
    Let \(e_i\) and \(v_i\) be as in Lemma \ref{D8representationEven}. We saw in the proof of Lemma \ref{D8representationEven} that \(e_1 + e_3 = W_K\) is an invariant of \(D_8\) and, hence, lies in \(\mathbb{C}[s,t]\). We also noted that \(e_2 - e_4\) is a semi-invariant of the same weight as \(D\) so, by Lemma \ref{D8representationEven}, \(e_2 - e_4\) lies in \(D\mathbb{C}[s,t]\). Now, \(p_1\) and \(p_2\) are instances of \(v_1\) and \(v_2\), respectively, and the same is true for \(q_1\) and \(q_2\), since
    \begin{align*}
        &q_1 = \frac{1}{2}\biggl[\biggl(W_{L_1}^-(x,y) + W_{L_1}^-(y,x)\biggr) - \biggl(W_{L_2}^-(x,y) + W_{L_2}^-(y,x)\biggr)\biggr] \\
        &q_2 = \frac{1}{2}\biggl[\biggl(W_{L_1}^+(x,y) - W_{L_1}^+(y,x)\biggr) - \biggl(W_{L_2}^+(x,y) - W_{L_2}^+(y,x)\biggr)\biggr].
    \end{align*}
    Thus, \(-p_2v_1 + p_1v_2\) and \(-q_2v_1 + q_1v_2\) are both semi-invariants of weight \(\chi\). Therefore, by Lemma \ref{semi-invariantD}, there are invariants \(f,g \in \mathbb{C}[s,t]\) such that \(-p_2v_1 + p_1v_2 = Df\) and \(-q_2v_1 + q_1v_2 = Dg\). This gives a linear system in \(v_1\) and \(v_2\) with determinant \(p_1q_2 - p_2q_1 = D\). Therefore, \(v_1 = q_1f - p_1g\) and \(v_2 = q_2f - p_2g\) as desired.
    
    To show that the sums are direct, we show that, in fact, \(\{1,\, D,\, p_1,\, p_2,\, q_1,\, q_2\}\) is linearly independent over \(\mathbb{C}[s,t]\). Suppose that 
    \begin{equation*}
        f_1 + p_2f_2 + q_2f_3 + Dg_1 + p_1g_2 + q_1g_3 = 0. 
    \end{equation*}
    The monomials in any polynomial from \(\mathbb{C}[s,t]\) have even powers in both \(x\) and \(y\) so it must be the case that 
    \begin{align*}
        &f_1 + p_2f_2 + q_2f_3 = 0 \\
        &Dg_2 + p_1g_2 + q_1g_3 = 0
    \end{align*}
    since the monomials in \(p_2\) and \(q_2\) also have even powers in \(x\) and \(y\) while the monomials in \(D,\,p_1\), and \(q_1\) have odd powers in \(x\) and \(y\). The first equation implies that \(p_2f_2 + q_2f_3\) is an invariant so
    \begin{align*}
        A \cdot (p_2f_2 + q_2f_3) &= p_1 f_2 + q_1 f_3 \\
        & = p_2f_2 + q_2f_3.
    \end{align*}
    Then, \((p_1 - p_2)f_2 = (q_2 - q_1)f_3\) and so
    \begin{align*}
        f_2 &= \frac{q_2 - q_1}{p_1 - p_2}f_3 \\
        \implies A \cdot f_2 &= \frac{q_1 + q_2}{-p_1 - p_2}f_3.
    \end{align*}
    Equating the above two equations, it follows that \(f_3 = f_2 = 0\) and, hence, \(f_1 = 0\).
    
    Similarly, we have that \(p_1g_2 + q_1g_3\) is a semi-invariant of weight \(\chi\) so
    \begin{align*}
        A \cdot (p_1g_2 + q_1g_3) &= -p_2 g_2 - q_2 g_3 \\
        & = p_1g_2 + q_1g_3.
    \end{align*}
    Just as in the previous case, it follows that \(g_1 = g_2 = g_3 = 0\).
\end{proof}

\begin{remark}
    Since the degree of \(D\) is \(8\), Theorem \ref{WEtheoremI} gives that \(W_L^-(x,y) = W_L^-(y,x)\) for \(n < 8\). In this case, \(W_L\) satisfies the ordinary MacWilliams identity, and the simplest example where this fails occurs for length \(8\) (see Example \ref{ex:odd_from_ext_Hamming}).
\end{remark}

\begin{corollary}
    With \(p_i\) and \(q_i\) as in Theorem \ref{WEtheoremI}, set \(p = p_1 + p_2\) and \(q = q_1 + q_2\). Then, \(W_L \in \mathbb{C}[s,t] \oplus D\mathbb{C}[s,t] \oplus p\mathbb{C}[s,t] \oplus q\mathbb{C}[s,t]\).
\end{corollary}

\subsubsection{Consequences for Lagrangians of Type II} We now impose the additional assumption that each \(v \in L^+\) has weight divisible by \(4\). If \(L \subset \mathbb{F}_2^n\) is even, then we saw in Section \ref{MWSection} that this implies \(n \equiv 0 \ (\textrm{mod}\ 8)\). However, this is not true for odd Lagrangians as Table \ref{Table:OddLag} indicates. We give two examples of such Lagrangians.

\begin{example}
    Let \(i_2' = \{00,\, 10\}\). Then \(i_2'^+ = \{00\}\) trivially satisfies our hypothesis yet \(n \equiv 2 \ (\textrm{mod}\ 8)\).
\end{example}

\begin{example}
    Consider the odd Lagrangian \(L_2\) from Example \ref{OddLagLength6} whose weight enumerator is \(W_{L_2}(x,y) = y^6 + 4x^3y^3 + 3x^4y^2\). Then, \(L_2^+\) satisfies our hypothesis, yet \(n \equiv -2 \ (\textrm{mod}\ 8)\).
\end{example}

It turns out that Lagrangians of type II only exist for \(n \equiv 0\ (\textrm{mod}\ 8),\, n \equiv 2 \ (\textrm{mod}\ 8)\), and \(n \equiv -2 \ (\textrm{mod}\ 8)\). We now proceed to show this and study the weight enumerators of such Lagrangians.

First, we observe that all odd vectors in \(L\) have the same weight modulo \(4\). Indeed, if \(\xi \in L\) is any odd vector, then Proposition \ref{prop:reduction_to_even} tells us that \(L\) admits the decomposition \(L = L^+ \oplus \langle \xi \rangle\). Hence, if \(\xi' = v + \xi \in L\) is any other odd vector, then \(\lVert \xi' \rVert = \lVert v \rVert + \lVert \xi \rVert - 2 \lVert v\xi \rVert\) where the product \(v \xi\) is taken component-wise. Since the parity of the integer \(\lVert v \xi \rVert\) is equal to \(v\cdot \xi = 0\), we get \(\lVert \xi' \rVert \equiv \lVert \xi \rVert \ (\textrm{mod}\ 4)\).

Recall from Section \ref{MWSection} that the group \(G\) generated by
\begin{equation*}
    A = \frac{1}{\sqrt{2}}
    \begin{bmatrix}
        1 & 1 \\
        -1 & 1
    \end{bmatrix}
    \textnormal{ and }
    B = 
    \begin{bmatrix}
        i & 0 \\
        0 & 1
    \end{bmatrix}
\end{equation*}
has order \(192\) and its algebra of invariants is \(\mathbb{C}[s,t]\) where \(s = W_{e_8} = y^8 + 14x^4y^4 + x^8\) and \(t = x^4y^4(x^4-y^4)^4\).

\begin{lemma}\label{repG_n=0(4)}
    Set \(e_1 = W_L^+(x,y),\, e_2 = W_L^-(x,y),\, e_3 = W_L^+(y,x),\, e_4 = W_L^-(y,x)\) and suppose that \(n \equiv 0 \ (\textrm{mod}\ 4)\). Then, \(V\) is \(G\)-invariant and its decompositions into irreducible \(G\)-submodules is given by \(V = \langle v_0 \rangle \oplus \langle v_1, v_2, v_3 \rangle\) where \(v_0 = e_1 + e_3,\, v_1 = e_2 + e_4,\, v_2 = e_1 - e_3\), and \(v_3 = e_2 - e_4\).
\end{lemma}

\begin{proof}
    Lemma \ref{D8representationEven} implies that \(V\) is invariant under the action by \(A\). Since all vectors in \(L^+\) have weights divisible by 4, we have \(B \cdot e_1 = e_1\). Moreover, \(B \cdot e_3 = e_3\) since \(n \equiv 0 \pmod{4}\) gives that \(n - \lVert v \rVert \equiv 0\ (\textrm{mod}\ 4)\) for every \(v \in L^+\). Hence, \(e_1+e_3\) is fixed by \(B\). We saw in Lemma \ref{D8representationEven} that \(e_1+e_3\) is fixed by \(A\), so it is an invariant of \(G\). 
    
    If all odd vectors have weight congruent to \(1\) modulo \(4\), then \(B \cdot e_2 = -i e_2\) and \(B \cdot e_4 = ie_4\) and, hence, \(B \cdot v_1 =-iv_3\) and \(B \cdot v_3 = -iv_1\). Otherwise, we have that \(B \cdot e_2 = ie_2\) and \(B \cdot e_4 = -ie_4\) and, hence, \(B \cdot v_1 = iv_3\) and \(B \cdot v_3 = iv_1\). We also saw in Lemma \ref{D8representationEven} that \(A \cdot v_1 = -v_2\), \(A \cdot v_2 = v_1\), and \(A \cdot v_3 = v_3\). Noting that \(v_2\) is fixed by \(B\) gives that \(\langle v_1, v_2, v_3 \rangle\) is a \(G\)-submodule. A computation shows that it contains no semi-invariants and is therefore irreducible.
\end{proof}

\begin{theorem}\label{W_L:n=0(8)}
    Suppose that \(n \equiv 0 \ (\mathrm{mod}\ 4)\) and that \(L \subset \mathbb{F}_2^n\) is a Lagrangian subspace of Type II. Set \(v_0 = W_L^+(x,y) + W_L^+(y,x),\, v_1 = W_L^-(x,y) + W_L^-(y,x),\, v_2 = W_L^+(x,y) - W_L^+(y,x)\), and \(v_3 = W_L^-(x,y) - W_L^-(y,x)\). Then, \(n \equiv 0 \ (\mathrm{mod}\ 8)\) and
    \begin{align*}
        &v_0 \in \mathbb{C}[s,t] \\
        &v_1 \in p(q_1r_3+q_3r_1)\mathbb{C}[s,t] \oplus (p_1r_3 - p_3r_1)\mathbb{C}[s,t] \oplus -pw(p_1q_3+p_3q_1)\mathbb{C}[s,t] \\
        &v_2 \in p(q_3r_2-q_2r_3)\mathbb{C}[s,t] \oplus (p_2r_3 - p_3r_2)\mathbb{C}[s,t] \oplus pw(p_3q_2-p_2q_1)\mathbb{C}[s,t] \\
        &v_3 \in -p(q_2r_1+q_1r_2)\mathbb{C}[s,t] \oplus (p_2r_1 - p_1r_2)\mathbb{C}[s,t] \oplus pw(p_2q_1+p_1q_2)\mathbb{C}[s,t].
    \end{align*}
    where
    \begin{equation*}
        \begin{bmatrix}
            p_1 \\
            p_2 \\
            p_3
        \end{bmatrix}
        =
        \begin{bmatrix}
            2xy \\
            y^2-x^2 \\
            y^2+x^2
        \end{bmatrix}, \:
        \begin{bmatrix}
            q_1 \\
            q_2 \\
            q_3
        \end{bmatrix}
        =
        \begin{bmatrix}
            p_2p_3 \\
            -p_1p_3 \\
            p_1p_2
        \end{bmatrix}, \:
        \begin{bmatrix}
            r_1 \\
            r_2 \\
            r_3
        \end{bmatrix}
        =
        \begin{bmatrix}
            p_1^3 \\
            p_2^3 \\
            -p_3^3
        \end{bmatrix},
    \end{equation*}
    and \(p = p_1p_2p_3 = 2xy(y^4-x^4)\) and \(w = x^2y^2(x^4-y^4)^2\) are as in Lemma \ref{semi-invGI}. 
\end{theorem}

\begin{proof}
    We saw in Lemma \ref{W_L:n=0(8)} that \(v_0\) is an invariant and hence lies in \(\mathbb{C}[s,t]\), implying that \(n = \deg v_0 \equiv 0 \ (\textrm{mod}\ 8)\). Assume first that all odd weights are congruent to \(1\) modulo \(4\). We saw in the proof of Lemma \ref{semi-invGI} that 
    \begin{align*}
        &A\cdot p_1 = -p_2,\, A\cdot p_2 = p_1,\, A\cdot p_3 = p_3 \\
        &B\cdot p_1 = -ip_1,\, B\cdot p_2 = p_3,\, B\cdot p_3 = p_2.
    \end{align*}
    It follows that \(p_2v_1 - p_1v_2 + p_3v_3\) is a semi-invariant of weight \(\chi_i\) (change \(i\) accordingly as in table) and, hence, belongs to \(pu\mathbb{C}[s,t]\). It also follows that \(G\) acts on each \(q_i\) and each \(r_i\) in the following way:
    \begin{align*}
        &A\cdot q_1 = -q_2,\, A\cdot q_2 = q_1,\, A\cdot q_3 = -q_3 \\
        &B\cdot q_1 = q_1,\, B\cdot q_2 = iq_3,\, B\cdot q_3 = iq_2 
    \end{align*}
    and
    \begin{align*}
        &A\cdot r_1 = -r_2,\, A\cdot r_2 = r_1,\, A\cdot r_3 = r_3 \\
        &B\cdot r_1 = ir_1,\, B\cdot r_2 = -r_3,\, B\cdot r_3 = -r_2. 
    \end{align*}
    This, along with Lemmas \ref{repG_n=0(4)} amd \ref{semi-invGI}, implies that there are invariants \(f, g, h \in \mathbb{C}[s,t]\) such that 
    \begin{equation*}
        \begin{bmatrix}
            p_2 & -p_1 & p_3 \\
            q_2 & q_1 & q_3 \\
            r_2 & -r_1 & r_3
        \end{bmatrix}
        \begin{bmatrix}
            v_1 \\
            v_2 \\
            v_3
        \end{bmatrix}
        =
        \begin{bmatrix}
            (pu)f \\
            ug \\
            (puw)h
        \end{bmatrix}
    \end{equation*}
    where \(u\) is as in Lemma \ref{semi-invGI}. Letting \(N\) be the coefficient matrix, it can be seen that \(\det N = 4u\), giving
    \begin{equation*}
        \begin{bmatrix}
            v_1 \\
            v_2 \\
            v_3
        \end{bmatrix}
        = \frac{1}{4}M
        \begin{bmatrix}
            pf \\
            g \\
            pwh
        \end{bmatrix} =
        \frac{1}{4} \tilde{M}
        \begin{bmatrix}
            f \\
            g \\ 
            h
        \end{bmatrix}
    \end{equation*}
    where
    \begin{equation*}
        M = 
        \begin{bmatrix}
            q_1r_3+q_3r_1 & p_1r_3 - p_3r_1 & -(p_1q_3+p_3q_1) \\
            q_3r_2-q_2r_3 & p_2r_3 - p_3r_2 & p_3q_2-p_2q_1 \\
            -(q_2r_1+q_1r_2) & p_2r_1 - p_1r_2 & p_2q_1+p_1q_2
        \end{bmatrix}
    \end{equation*}
    is the adjugate of \(N\) and
    \begin{equation*}
        \tilde{M} = M
        \begin{bmatrix}
            p & 0 & 0 \\
            0 & 1 & 0 \\
            0 & 0 & pw
        \end{bmatrix},
    \end{equation*}
    so the \(i\)th row of \(\tilde{M}\) consists of the generators of the module that \(v_i\) is claimed to belong to. It remains to check that these generators are free.
    
    Since \(A \cdot \tilde{M}_{1j} = -\tilde{M}_{2j}\) and \(B \cdot \tilde{M}_{1j} = -i\tilde{M}_{3j}\) for \(j = 1,2,3\), it suffices to show that \(\tilde{M}_{11},\, \tilde{M}_{12}\), and \(\tilde{M}_{13}\) are linearly independent over \(\mathbb{C}[s,t]\). If \(f_1, f_2, f_3 \in \mathbb{C}[s,t]\) are such that \(\tilde{M}_{11}f_1 + \tilde{M}_{12}f_2 + \tilde{M}_{13}f_3 = 0 \), then the action of \(A\) and \(B\) on the first row of \(\tilde{M}\) implies that 
    \begin{equation*}
        \tilde{M}
        \begin{bmatrix}
            f_1 \\ 
            f_2 \\ 
            f_3
        \end{bmatrix}
        = 0.
    \end{equation*}
    Since \(\det \tilde{M} = p^2w \det M = (2w)^2 (\det N)^2 = (8wu)^2 \neq 0\), we conclude that \(f_i = 0\).
    
    For the case when all odd weights are congruent to \(3\) modulo \(4\), simply replace \(p_3\) by \(-p_3\) and apply the same argument.
\end{proof}

\begin{lemma}\label{repG:n=2(4)}
    Set \(e_1 = W_L^+,\, e_2 = W_L^-,\, e_3(x,y) = W_L^+(y,x),\, e_4 = W_L^-(y,x)\) and suppose that \(n \equiv 2 \ (\textrm{mod}\ 4)\). Then, \(V\) is \(G\)-invariant and its decomposition into irreducible \(G\)-submodules is given by \(V = \langle v_0 \rangle \oplus \langle v_1, v_2, v_3 \rangle\) where \(v_0 = e_2-e_4,\, v_1 = e_2 + e_4,\, v_2 = e_1 - e_3\), and \(v_3 = e_1 + e_3\).
\end{lemma}

\begin{proof}
    The proof is similar to Lemma \ref{repG_n=0(4)}. We only note that, if all odd weights are congruent to \(1\) modulo \(4\), then \(B \cdot e_2 = -i e_2\) and \(B \cdot e_4 = -ie_4\), so \(e_2 - e_4\) is a semi-invariant whose weight \(\chi\) maps \(A \mapsto 1\) and \(B \mapsto -i\). Also, \(B\cdot v_1 = -iv_1,\, B\cdot v_2 = v_3\), and \(B\cdot v_3 = v_2\).
    
    If all odd weights are congruent to \(3\) modulo 4, then the weight of the semi-invariant \(e_2-e_4\) is the character \(\chi^{-1}\) that maps \(A \mapsto 1\) and \(B \mapsto i\). Also, \(B\cdot v_1 = iv_1,\, B\cdot v_2 = v_3\), and \(B\cdot v_3 = v_2\).
\end{proof}

% put polynomials into table with columns i, ui, pi, qi, ri
\begin{theorem}\label{W_L:n=2(8)}
    Suppose that \(n \equiv 2 \ (\textrm{mod}\ 4),\, L\subset \mathbb{F}_2^n\) is a Lagrangian subspace of Type II, and that all odd weights in \(L\) are congruent to \(1\) modulo \(4\). Set \(v_0 = W_L^-(x,y) - W_L^-(y,x),\, v_1 = W_L^-(x,y) + W_L^-(y,x),\, v_2 = W_L^+(x,y) - W_L^+(y,x)\), and \(v_3 = W_L^+(x,y) + W_L^+(y,x)\). Then, \(n \equiv 2 \pmod{8}\) and 
    \begin{align*}
        &v_0 \in pu\mathbb{C}[s,t] \\
        &v_1 \in u_1\mathbb{C}[s,t] \oplus p(p_3r_1 - p_2r_3)\mathbb{C}[s,t] \oplus (p_2q_3-p_3q_2)\mathbb{C}[s,t] \\
        &v_2 \in u_2\mathbb{C}[s,t] \oplus p(p_3r_2 - p_1r_3)\mathbb{C}[s,t] \oplus (p_1q_3-p_3q_1)\mathbb{C}[s,t] \\
        &v_3 \in u_3\mathbb{C}[s,t] \oplus p(p_1r_1 - p_2r_2)\mathbb{C}[s,t] \oplus (p_2q_1-p_1q_2)\mathbb{C}[s,t].
    \end{align*}
    where \(u_i,\, p_i,\, q_i\), and \(r_i\) are given by the following table:
    \begin{center}
        \begin{tabular}{|c|c|c|c|c|}
             \hline
             \(i\) & \(u_i\) & \(p_i\) & \(q_i\) & \(r_i\) \\ \hline
             \(1\) &  \(2xy\) & \(-u_2u_3\) & \(4xy(y^4+x^4)\) & \(xy^7 + 7x^5y^3 + x^7y + 7x^3y^5\) \\
             \(2\) &  \(y^2 - x^2\) & \(u_1u_3\) & \(y^6 + 5x^2y^4 - 5x^4y^2 - x^6\) & \(y^8 - x^8\) \\
             \(3\) &  \(y^2 + x^2\) & \(u_1u_2\) & \(y^6 - 5x^2y^4 - 5x^4y^2 + x^6\) & \(xy^7 + 7x^5y^3 - x^7y - 7x^3y^5\) \\
             \hline
        \end{tabular}
    \end{center}
    \iffalse
    \begin{align*}
        \begin{bmatrix}
            u_1 \\
            u_2 \\
            u_3
        \end{bmatrix}
        =
        \begin{bmatrix}
             2xy \\
            y^2 - x^2 \\
            y^2 + x^2
        \end{bmatrix}, \:
        \begin{bmatrix}
            p_1 \\
            p_2 \\
            p_3
        \end{bmatrix}
        &=
        \begin{bmatrix}
            -u_2u_3 \\
            u_1u_3 \\
            u_1u_2
        \end{bmatrix}, \:
        \begin{bmatrix}
            q_1 \\
            q_2 \\
            q_3
        \end{bmatrix}
        =
        \begin{bmatrix}
            4xy(y^4+x^4) \\
            y^6 + 5x^2y^4 - 5x^4y^2 - x^6 \\
            y^6 - 5x^2y^4 - 5x^4y^2 + x^6
        \end{bmatrix}, \\
        \begin{bmatrix}
            r_1 \\
            r_2 \\
            r_3
        \end{bmatrix}
        &=
        \begin{bmatrix}
            xy^7 + 7x^5y^3 + x^7y + 7x^3y^5 \\
            y^8 - x^8 \\
            xy^7 + 7x^5y^3 - x^7y - 7x^3y^5
        \end{bmatrix},
    \end{align*}
    \fi
    and \(p = u_1u_2u_3\) and \(u = -\frac{1}{2}(x^{12} - 33x^8y^4 - 33x^4y^8 + y^{12})\) are as in Lemma \ref{semi-invGI}.
\end{theorem}

\begin{proof}
    We saw in the proof of Lemma \ref{repG:n=2(4)} that \(e_2 - e_4\) is a semi-invariant of the same weight as \(pu\) and, hence, by Lemma \ref{semi-invGI}, lies in \(pu \mathbb{C}[s,t]\). Furthermore, one can check that \(A\) acts on \(p_i, q_i, r_i\) as on \(v_i\) except that \(A \cdot p_3 = -p_3\) and \(A\cdot q_3 = -q_3\). Moreover,
    \begin{align*}
        &B\cdot p_1 = p_1,\, B\cdot p_2 = -ip_3,\, B\cdot p_3 = -ip_2, \\
        &B\cdot q_1 = -iq_1,\, B\cdot q_2 = -q_3,\, B\cdot q_3 = -q_2 \\
        &B\cdot r_1 = -ir_3,\, B\cdot r_2 = r_2,\, B\cdot r_3 = -ir_1.
    \end{align*}
    This, along with Lemmas \ref{repG:n=2(4)} and \ref{semi-invGI}, implies that there are invariants \(f, g, h \in \mathbb{C}[s,t]\) such that 
    \begin{equation*}
        \begin{bmatrix}
            -p_1 & p_2 & p_3 \\
            -q_1 & q_2 & q_3 \\
            -r_2 & r_1 & r_3
        \end{bmatrix}
        \begin{bmatrix}
            v_1 \\
            v_2 \\
            v_3
        \end{bmatrix}
        =
        \begin{bmatrix}
            pf \\
            (uw)g \\
            (pu)h
        \end{bmatrix}
    \end{equation*}
    where \(w\) is as in Lemma \ref{semi-invGI}. Letting \(N\) be the coefficient matrix, it can be seen that \(\det N = -6pu\). Recalling that \(p^2 = 4w\), we have
    \begin{equation*}
        \begin{bmatrix}
            v_1 \\
            v_2 \\
            v_3
        \end{bmatrix}
        = -\frac{1}{6}M
        \begin{bmatrix}
            f/u \\
            pg/4 \\
            h
        \end{bmatrix} =
        -\frac{1}{6} \tilde{M}
        \begin{bmatrix}
            f \\
            g \\
            h
        \end{bmatrix}
    \end{equation*}
    where \(M\) is the adjugate of \(N\) and
    \begin{equation*}
        \tilde{M} = M
        \begin{bmatrix}
            1/u & 0 & 0 \\
            0 & p/4 & 0 \\
            0 & 0 & 1
        \end{bmatrix}.
    \end{equation*}
    A straightforward calculation shows that \(\tilde{M}_{j1} = -2u_j\) and, therefore, the \(i\)th row of \(\tilde{M}\) consists of the generators of the module that \(v_i\) is claimed to belong to. 
    
    For freeness, since \(A\cdot \tilde{M}_{2j} = \tilde{M}_{1j}\) and \(B \cdot \tilde{M}_{2j} = \tilde{M}_{3j}\) for \(j = 1,2,3\), it suffices to show that \(\tilde{M}_{21},\, \tilde{M}_{22}\), and \(\tilde{M}_{23}\) are linearly independent over \(\mathbb{C}[s,t]\), which is achieved by the same argument as in Theorem \ref{W_L:n=0(8)}.

    Since \(\deg \tilde{M}_{ij} \equiv 2 \pmod{8}\) for \(i,j = 1,2,3\), we have \(n = \deg v_i \equiv 2 \pmod{8}\).
\end{proof}

\begin{remark}
    Since the degree of \(pu\) in Theorem \ref{W_L:n=2(8)} is 18, we find that, for any Lagrangian \(L\) of length \(n < 18\) satisfying the hypothesis of the theorem, we have \(W_L^-(x,y) = W_L^-(y,x)\). 
\end{remark}

\begin{theorem}\label{W_L:n=-2(8)} 
    Suppose that \(n \equiv 2 \ (\textrm{mod}\ 4),\, L\subset \mathbb{F}_2^n\) is a Lagrangian subspace of Type II and that all odd weights in \(L\) are congruent to \(3\) modulo \(4\). Set \(v_0 = W_L^-(x,y) - W_L^-(y,x),\, v_1 = W_L^-(x,y) + W_L^-(y,x),\, v_2 = W_L^+(x,y) - W_L^+(y,x)\), and \(v_3 = W_L^+(x,y) + W_L^+(y,x)\). Then, \(n \equiv -2 \pmod{8}\) and 
    \begin{align*}
        &v_0 \in puw\mathbb{C}[s,t] \\
        &v_1 \in (q_2r_3 + q_3r_1)\mathbb{C}[s,t] \oplus q_1\mathbb{C}[s,t] \oplus u(p_2q_3-p_3q_2)\mathbb{C}[s,t] \\
        &v_2 \in (q_3r_2-q_1r_3)\mathbb{C}[s,t] \oplus q_2\mathbb{C}[s,t] \oplus u(p_1q_3+p_3q_1)\mathbb{C}[s,t] \\
        &v_3 \in -(q_1r_1+q_2r_2)\mathbb{C}[s,t] \oplus q_3\mathbb{C}[s,t] \oplus -u(p_1q_2+p_2q_1)\mathbb{C}[s,t].
    \end{align*}
    where \(u_i,\, p_i,\, q_i\), and \(r_i\) are given by the following table:
    \begin{center}
        \begin{tabular}{|c|c|c|c|c|}
             \hline
             \(i\) & \(u_i\) & \(p_i\) & \(q_i\) & \(r_i\) \\ \hline
             \(1\) &  \(2xy\) & \(u_2u_3\) & \(u_1^3\) & \(xy^7 + 7x^5y^3 + x^7y + 7x^3y^5\) \\
             \(2\) &  \(y^2 - x^2\) & \(-u_1u_3\) & \(u_2^3\) & \(y^8 - x^8\) \\
             \(3\) &  \(y^2 + x^2\) & \(u_1u_2\) & \(-u_3^3\) & \(xy^7 + 7x^5y^3 - x^7y - 7x^3y^5\) \\
             \hline
        \end{tabular}
    \end{center}
    \iffalse
    \begin{align*}
        \begin{bmatrix}
            u_1 \\
            u_2 \\
            u_3
        \end{bmatrix}
        &=
        \begin{bmatrix}
             2xy \\
            y^2 - x^2 \\
            y^2 + x^2
        \end{bmatrix}, \:
        \begin{bmatrix}
            p_1 \\
            p_2 \\
            p_3
        \end{bmatrix}
        =
        \begin{bmatrix}
            u_2u_3 \\
            -u_1u_3 \\
            u_1u_2
        \end{bmatrix}, \:
        \begin{bmatrix}
            q_1 \\
            q_2 \\
            q_3
        \end{bmatrix}
        =
        \begin{bmatrix}
            u_1^3 \\
            u_2^3 \\
            -u_3^3
        \end{bmatrix}, \\
        \begin{bmatrix}
            r_1 \\
            r_2 \\
            r_3
        \end{bmatrix}
        &=
        \begin{bmatrix}
            xy^7 + 7x^5y^3 + x^7y + 7x^3y^5 \\
            y^8 - x^8 \\
            xy^7 + 7x^5y^3 - x^7y - 7x^3y^5
        \end{bmatrix}
    \end{align*}
    \fi
    and \(p = u_1u_2u_3\) and \(u = -\frac{1}{2}(x^{12} - 33x^8y^4 - 33x^4y^8 + y^{12})\) are as in Lemma \ref{semi-invGI}.
\end{theorem}

\begin{proof}
    We saw in the proof of Lemma \ref{repG:n=2(4)} that \(e_2 - e_4\) is a semi-invariant of the same weight as \(puw\) and hence, by Lemma \ref{semi-invGI}, lies in \(puw\mathbb{C}[s,t]\). Furthermore, \(G\) acts on \(p_i,\, q_i,\) and \(r_i\) as in the proof of Theorem \ref{W_L:n=2(8)} except that
    \begin{equation*}
        B \cdot p_2 = ip_3, \: B \cdot p_3 = ip_2, \: B \cdot q_1 = iq_1.
    \end{equation*}
    This, along with Lemmas \ref{repG:n=2(4)} and \ref{semi-invGI}, implies that there are invariants \(f,g,h \in \mathbb{C}[s,t]\) such that
    \begin{equation*}
        \begin{bmatrix}
            -p_1 & p_2 & p_3 \\
            q_1 & q_2 & q_3 \\
            r_2 & -r_1 & r_3
        \end{bmatrix}
        \begin{bmatrix}
            v_1 \\
            v_2 \\
            v_3
        \end{bmatrix}
        =
        \begin{bmatrix}
            (pw)f \\
            wg \\
            (puw)h
        \end{bmatrix}
    \end{equation*}
    where \(w\) is as in Lemma \ref{semi-invGI}. Letting \(N\) be the coefficient matrix, it can be shown that \(\det N = 18pw\) and thus
    \begin{equation*}
        \begin{bmatrix}
            v_1 \\
            v_2 \\
            v_3
        \end{bmatrix}
        = \frac{1}{18}M
        \begin{bmatrix}
            f \\
            g/p \\
            uh
        \end{bmatrix}
    \end{equation*}
    where \(M\) is the adjugate of \(N\). Define
    \begin{equation*}
        \tilde{M} = M
        \begin{bmatrix}
            1 & 0 & 0 \\
            0 & 1/p & 0 \\
            0 & 0 & u
        \end{bmatrix}.
    \end{equation*}
    A straightforward calculation shows that \(\tilde{M}_{j2} = -\frac{2}{3}q_j\) for \(j = 1,2\) and that \(\tilde{M}_{32} = \frac{2}{3}q_3\). Thus, the \(i\)th row of \(\tilde{M}\) consists of the generators of the module that \(v_i\) is claimed to belong to. The rest of the proof is completed in the same manner as in Theorem \ref{W_L:n=2(8)}.
\end{proof}

\begin{remark}
    Since the degree of \(puw\) in Theorem \ref{W_L:n=-2(8)} is 30, we find that, for any Lagrangian \(L\) of length \(n < 30\) satisfying the hypothesis of the theorem, we have \(W_L^-(x,y) = W_L^-(y,x)\).
\end{remark}

\appendix

\section{Code} \label{code}

The following code, written using SageMath \cite{sage}, is used to compute the statistics for odd Lagrangians found in Table \ref{Table:OddLag}. The list of generator matrices for representatives of inequivalent classes of self-dual codes can be obtained from \cite{database}.

\begin{verbatim}

import time
start = time.time()

def complementsOfOnes(n, G, ones): 
    L0 = G[1:] 
    complements = [] 
    for a in VectorSpace(GF(2), n-1):
        comp = copy(L0)
        for i in range(n-1):
            comp[i] += a[i] * ones
        K = LinearCode(comp)
        flag = True
        for j in range(len(complements)):
            if complements[j][0].is_permutation_equivalent(K):
                flag = False
                break
        if flag:
            complements.append([LinearCode(K), comp])
    return complements

# list of generator matrices for inequivalent self-dual codes
listSD = []

n = len(listSD[0]) # dimension
MS = MatrixSpace(GF(2), n, 2*n)
# Matrix representing our inner product wrt standard basis
J = matrix(GF(2), 2*n, 2*n, lambda i,j: 0 if i == j else 1) 
SD = [MS(m) for m in listSD]

weightDist = []
total = 0
print("Inequivalent odd Lagrangians of length", 2*n, "\n")
for G in SD:
    for L in complementsOfOnes(n, G, ones):
        M = L[1]*J
        for v in M.right_kernel_matrix():
            if sum(v) == 1:
                L1 = L[1].stack(v)
                L2 = L[1].stack(ones + v)
                break
        total += 1
        print("#", total)
        print(L1)
        C1 = LinearCode(L1)
        C2 = LinearCode(L2)
        dist = C1.spectrum()
        weightDist.append(dist)
        print("Weight distribution: ", dist, "\n")
        if not C1.is_permutation_equivalent(C2):
            total += 1
            print("#", total)
            print(L2)
            dist = C2.spectrum()
            weightDist.append(dist)
            print("Weight distribution: ", dist, "\n")

print("Total number of inequivalent odd Lagrangians: ", total)
print()

minDist = 0
for i in range(total):
    d = weightDist[i]
    flag = True
    for j in range(1, 2*n):
        if d[j] != 0:
            if j <= minDist:
                break
            else:
                minDist = j
                break
    for j in range(2, 2*n, 4):
        if d[j] != 0:
            flag = False
            break
    if flag:
        print("Even weights divisible by 4 given by #", i+1)
        
print()
print("Greatest minimum distance:", minDist)
print()
print("Lagrangians which achieve the greatest minimum distance:")
numDist = 0
for i in range(total):
    d = weightDist[i]
    for j in range(1, 2*n):
        if d[j] != 0:
            if j == minDist:
                print("#", i+1)
                numDist += 1
            else:
                break

print("Time for execution:", time.time() - start, "seconds")
\end{verbatim}

\end{document}